\PassOptionsToPackage{hidelinks}{hyperref}
\documentclass[lettersize,journal]{IEEEtran}
\IEEEaftertitletext{\vspace{-1.8\baselineskip}}
\usepackage{algorithm}
\usepackage{algpseudocode}
\usepackage{amsmath,amsfonts,amssymb}
\usepackage{amsthm}
\usepackage{mathtools}
\usepackage{bm}
\usepackage{dsfont}
\usepackage{mathrsfs}
\usepackage{enumitem}
\usepackage{array}
\usepackage{booktabs}
\usepackage{multirow}
\usepackage{graphicx}
\usepackage[caption=false,font=normalsize,labelfont=sf,textfont=sf]{subfig}
\usepackage{stfloats}
\usepackage{textcomp}
\usepackage{url}
\usepackage{verbatim}
\usepackage{comment}
\usepackage{xcolor}
\usepackage{soul}
\usepackage{microtype}
\usepackage{orcidlink}
\usepackage{cite}

\usepackage[compact]{titlesec}

\titlespacing*{\section}
{0pt}                  
{1.0ex plus 0.3ex minus 0.2ex}  
{0.7ex plus 0.2ex}              

\titlespacing*{\subsection}
{0pt}
{0.8ex plus 0.2ex minus 0.2ex}
{0.45ex plus 0.15ex}

\titlespacing*{\subsubsection}
{0pt}
{0.6ex plus 0.2ex minus 0.2ex}
{0.35ex plus 0.1ex}

\usepackage{tikz}
\usetikzlibrary{matrix}
\usetikzlibrary{snakes,arrows,shapes}
\usetikzlibrary{arrows.meta,positioning}
\usetikzlibrary{calc,backgrounds}

\definecolor{ink}{HTML}{272726}     
\definecolor{infill}{HTML}{DADAD9}  
\colorlet{stateone}{red!30}         
\colorlet{statezero}{blue!30}       
\colorlet{onearrow}{red!60!black}   
\colorlet{zeroarrow}{blue!60!black} 
\colorlet{netbg}{black!7}           
\colorlet{netedge}{black!32}
\colorlet{meshgray}{black!45}       

\newcommand{\carglyph}[3]{%
\begin{scope}[shift={(#1)}, line cap=round, line join=round]
  \draw[fill=#2, draw=ink, line width=0.7pt]
        (-0.71,-0.11)
        .. controls (-0.73,-0.02) and (-0.69, 0.09) .. (-0.54, 0.125)
        .. controls (-0.44, 0.145) and (-0.37, 0.163) .. (-0.31, 0.181)
        .. controls (-0.26, 0.228) and (-0.23, 0.267) .. (-0.16, 0.270)
        .. controls (-0.08, 0.283) and ( 0.09, 0.283) .. ( 0.17, 0.271)
        .. controls ( 0.23, 0.250) and ( 0.25, 0.226) .. ( 0.29, 0.205)
        .. controls ( 0.35, 0.176) and ( 0.39, 0.162) .. ( 0.43, 0.153)
        .. controls ( 0.50, 0.143) and ( 0.55, 0.136) .. ( 0.60, 0.128)
        .. controls ( 0.67, 0.117) and ( 0.71, 0.070) .. ( 0.71,-0.01)
        .. controls ( 0.715,-0.06) and ( 0.71,-0.09) .. ( 0.70,-0.11)
        -- cycle;
  \node[font=\scriptsize\bfseries, text=ink, inner sep=0pt] at (-0.02,0.04) {#3};
  \fill[netbg] (-0.469,-0.11) circle (0.143);
  \fill[netbg] ( 0.398,-0.11) circle (0.143);
  \draw[fill=ink, draw=ink, line width=0.4pt] (-0.469,-0.11) circle (0.105);
  \draw[fill=ink, draw=ink, line width=0.4pt] ( 0.398,-0.11) circle (0.105);
\end{scope}}

\usepackage{siunitx}

\usepackage[capitalize,nameinlink]{cleveref}
\newtheorem{theorem}{Theorem}
\newtheorem{lemma}{Lemma}
\newtheorem{corollary}{Corollary}

\newtheorem{conjecture}{Conjecture}
\newtheorem{prop}{Proposition}
\newtheorem{assumption}{Assumption}
\theoremstyle{remark}

\crefname{theorem}{Theorem}{Theorems}
\Crefname{theorem}{Theorem}{Theorems}

\crefname{lemma}{Lemma}{Lemmas}
\Crefname{lemma}{Lemma}{Lemmas}

\crefname{corollary}{Corollary}{Corollaries}
\Crefname{corollary}{Corollary}{Corollaries}

\crefname{definition}{Definition}{Definitions}
\Crefname{definition}{Definition}{Definitions}

\crefname{conjecture}{Conjecture}{Conjectures}
\Crefname{conjecture}{Conjecture}{Conjectures}

\crefname{prop}{Proposition}{Propositions}
\Crefname{prop}{Proposition}{Propositions}

\crefname{assumption}{Assumption}{Assumptions}
\Crefname{assumption}{Assumption}{Assumptions}

\crefname{remark}{Remark}{Remarks}
\Crefname{remark}{Remark}{Remarks}

\DeclareMathOperator{\sgn}{sgn}

\DeclareMathOperator{\diag}{diag}

\DeclareMathOperator*{\argmin}{\arg\!\min}

\newcommand{\fvec}{\mathbf{f}}
\newcommand{\hvec}{\mathbf{h}}
\newcommand{\Slack}{\mathrm{Slack}}
\newcommand{\UR}{U_R}
\newcommand{\US}{U_S}
\newcommand{\bmat}[1]{\begin{bmatrix}#1\end{bmatrix}}
\newcommand{\one}{\mathbf{1}}
\newcommand{\eone}{\mathbf{e}_1}
\newcommand{\etwo}{\mathbf{e}_2}

\newcommand{\E}{\mathbb{E}}

\newcommand{\ind}{\mathds{1}}

\newcommand{\limt}{\lim_{t \to \infty}}

\newcommand{\reals}{\mathbb{R}}

\begin{document}

\title{Strategic Information Transmission over Gossip Networks}

\author{%
\IEEEauthorblockN{%
Emirhan Tekez\orcidlink{0009-0001-6615-3240}, \IEEEmembership{Graduate Student Member,~IEEE,}
Melih Bastopcu\orcidlink{0000-0001-5122-0642}, \IEEEmembership{Member,~IEEE,}\\
Sinan Gezici\orcidlink{0000-0002-6369-3081}, \IEEEmembership{Fellow,~IEEE}%
}%
\thanks{The authors are with the Department of Electrical and Electronics Engineering, Bilkent University, Ankara, Türkiye
(e-mails: emirhan.tekez@bilkent.edu.tr; bastopcu@bilkent.edu.tr; gezici@ee.bilkent.edu.tr).}%
\thanks{A part of this paper will be presented at the 65th IEEE Conference on Decision and Control (CDC), HI, USA, December 15--18, 2026 \cite{tekez2026strategic}.}%

\IEEEcompsocitemizethanks{
This work was supported by the TÜBİTAK 2232-B program (Project No. 124C533) and the TÜBİTAK 2210-A
program.}
}
\maketitle

\begin{abstract}
We consider a fully connected gossip network of $n$ nodes that track a binary continuous-time Markov source through a strategic sender transmitting updates under a communication budget at a rate that depends on the source state. The receivers exchange packets through gossip and decide whether to follow the sender. We model this interaction as a Stackelberg game and analyze it through a stochastic hybrid systems (SHS) framework. We prove that the sender's budget constraint binds at every interior optimum, reducing its problem to a one-dimensional search on the budget line. When the sender pushes its preferred state at the higher rate, the receivers gossip at the highest available rate. Gossip has no direction of its own and works against the asymmetry in the sender's policy rather than reinforcing it. We prove that an optimistic Stackelberg equilibrium exists, and that it is unique and explicitly characterized whenever a policy on the strategic half of that line is feasible at the gossip cap. Monte Carlo simulations agree with the analytical recursion.
\end{abstract}


\section{Introduction}\label{sec:intro}
\vspace{-0.1cm}
 
\IEEEPARstart{M}{odern} networked systems such as autonomous vehicle fleets, wireless sensor networks, and smart factories rely on the timely dissemination of rapidly changing status information. A widely used metric for quantifying information freshness is the age of information (AoI), which measures the time elapsed since the generation of the most recently received update \cite{kaul2012real,yates2020age_survey,sun2022age_book}. Although AoI captures timeliness, it does not directly capture whether the information currently held by a receiver is correct. A recent update may become inaccurate immediately after a source transition, whereas an older update may still be correct if the source has reverted to a previous state. This limitation has motivated accuracy-aware monitoring metrics such as the age of incorrect information (AoII) and version-based views of Markov source monitoring \cite{maatouk2020aoii,salimnejad2025age_versions}.
 
Additionally, in many settings, the entity that generates and transmits status updates does not share the same objective as the entities that consume them. A ride-sharing platform, for instance, may prefer drivers to believe that demand is high in certain zones so that they reposition accordingly, even as the actual demand state fluctuates. Drivers who share demand information through informal peer-to-peer channels want an accurate picture, not a strategically filtered one. In such settings the sender controls when to communicate, which creates a strategic dimension that timeliness-driven models do not capture, while gossip among the receivers can amplify, dilute, or correct what the sender transmits.
 
This work brings together ideas from gossip-based information dissemination and strategic information transmission by studying strategic persuasion in a timeliness-based gossip network. We consider a sender that observes a binary continuous-time Markov chain (CTMC) source and transmits version-stamped updates to a fully connected network of $n$ receivers, as shown in Fig.~\ref{fig:system_model}. Unlike standard monitoring models, the sender is strategic and allocates its limited communication budget across state-dependent push rates, sending updates at rate $s$ in state~$1$ and at rate $c$ in state~$0$, subject to the constraint $s + c \le R$. The receivers collaborate through gossiping at rate $\lambda\in[0,\Lambda]$, accepting only fresher packets and deciding whether it is worthwhile to participate at all. The sender prefers the network to hold state-$1$ information as often as possible, but it must still provide enough useful information to satisfy the receivers' participation requirement. We model this interaction as a Stackelberg game in which the sender commits to a policy $(s,c)$, after which the receiver team decides whether to follow the sender’s messages and selects a common gossip rate $\lambda$.
\begin{figure}[t]
\centering
\resizebox{\columnwidth}{!}{%
\begin{tikzpicture}[
  >=Latex, line cap=round, line join=round,
  font=\footnotesize, text=ink,
  statenode/.style={circle, draw=ink, line width=1.0pt, minimum size=0.94cm,
                    inner sep=0pt, font=\small\bfseries, text=ink},
  flow/.style={->, line width=0.9pt, draw=ink},
  meshline/.style={draw=meshgray, line width=0.5pt},
  lbl/.style={font=\footnotesize, text=ink, inner sep=1.5pt},
  rate/.style={font=\normalsize, inner sep=2pt}
]

\node[statenode, fill=stateone]  (S1) at (1.15, 1.00) {1};
\node[statenode, fill=statezero] (S0) at (1.15,-1.00) {0};
\draw[flow] (S0) to[bend left=40] node[lbl, left,  pos=0.5] {$q_{01}$} (S1);
\draw[flow] (S1) to[bend left=40] node[lbl, right, pos=0.5] {$q_{10}$} (S0);
\node[lbl] at (1.15,-1.95) {source $Q(t)$};

\draw[flow] (2.24,0) -- node[lbl, above] {observes $Q(t)$} (4.56,0);
\begin{scope}[shift={(5.35,0)}]
  \draw[fill=infill, draw=ink, line width=1.5pt] (0,0.72) circle (0.34);
  \draw[fill=infill, draw=ink, line width=1.5pt]
        (-0.60,-0.68)
        .. controls (-0.60, 0.12) and (-0.34, 0.30) .. ( 0.00, 0.30)
        .. controls ( 0.34, 0.30) and ( 0.60, 0.12) .. ( 0.60,-0.68)
        .. controls ( 0.60,-0.74) and ( 0.56,-0.76) .. ( 0.50,-0.76)
        -- (-0.50,-0.76)
        .. controls (-0.56,-0.76) and (-0.60,-0.74) .. (-0.60,-0.68)
        -- cycle;
\end{scope}
\node[lbl] at (5.35,-1.22) {sender};

\coordinate (NC) at (14.20,0);
\begin{scope}[on background layer]
  \filldraw[fill=netbg, draw=netedge, line width=1.4pt]
           (NC) ellipse [x radius=2.66cm, y radius=2.40cm];
\end{scope}

\foreach \k/\ang in {1/150, 2/90, 3/30, 4/-30, 5/-90, 6/-150}
  {\node[minimum width=1.50cm, minimum height=0.56cm, inner sep=0pt]
        (r\k) at ($(NC)+(\ang:1.88cm)$) {};}

\foreach \a/\b in {1/2,1/3,1/4,1/5,1/6,2/3,2/4,2/5,2/6,3/4,3/5,3/6,4/5,4/6,5/6}
  {\draw[meshline] (r\a) -- (r\b);}

\carglyph{r1}{statezero}{0}
\carglyph{r2}{stateone}{1}
\carglyph{r3}{stateone}{1}
\carglyph{r4}{stateone}{1}
\carglyph{r5}{statezero}{0}
\carglyph{r6}{stateone}{1}

\draw[flow, draw=onearrow,  line width=1.2pt] (6.08, 0.34) to[bend left=11]
      node[rate, text=onearrow, above, pos=0.55] {$\tfrac{s}{n}$} (11.67, 0.78);
\draw[flow, draw=zeroarrow, line width=1.2pt] (6.08,-0.34) to[bend right=11]
      node[rate, text=zeroarrow, below, pos=0.55] {$\tfrac{c}{n}$} (11.67,-0.78);
\end{tikzpicture}
}
\vspace{-0.35cm}
\caption{A binary CTMC source drives a strategic sender that pushes updates to each of the $n=6$ receivers at rate $\tfrac{s}{n}$ in state~$1$ and at rate $\tfrac{c}{n}$ in state~$0$, under the budget $s+c\le R$. Every pair of receivers is connected, and each ordered pair gossips at rate $\tfrac{\lambda}{n-1}$ and keeps the fresher packet. The digit on each vehicle is the state stored at that receiver's current packet.}
\label{fig:system_model}\vspace{-0.3cm}
\end{figure}
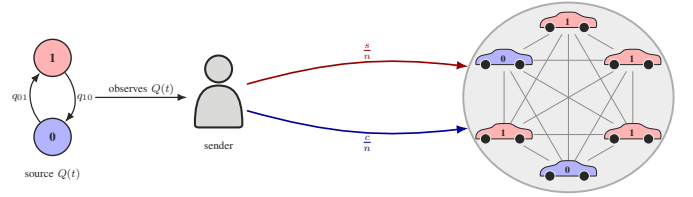

\subsection{Related Work}\label{sec:related}
\vspace{-0.1cm}
 
In large decentralized networks, dissemination often takes place through gossiping, where nodes exchange information peer-to-peer without centralized coordination \cite{boyd2005gossip,shah2008gossip}. Within the AoI literature, timely gossip has been studied over a variety of network models and topologies \cite{yates2021age_gossip,buyukates2022version,srivastava2023age,kaswan2025age_survey}. More recently, the literature has examined how misinformation spreads through gossip networks \cite{kaswan2025misinformation,maranzatto2025information}, as well as how state accuracy evolves in timeliness-based gossiping systems driven by Markov sources \cite{tekez2026accuracy}. Despite these advances, the existing gossip literature is largely non-strategic:
the source-side update process is typically exogenous, and the sender is not
modeled as an agent whose objectives may differ from the receivers'. 
 
A complementary line of work comes from strategic information transmission (SIT) and Bayesian persuasion. Strategic communication without commitment is studied through cheap-talk equilibria \cite{crawford1982strategic}, whereas Bayesian persuasion studies a sender that commits to an information structure \cite{kamenica2011bayesian}. This literature has been extended in several dynamic directions, including continuous-time and Markovian persuasion models \cite{ely2017beeps,ashkenazi2023markovian,aid2025continuous,lehrer2025markovian}. Persuasion over networks has also attracted growing interest, with prior work examining how network structure and public signaling shape the impact of strategic information \cite{egorov2020persuasion,candogan2019persuasion}. Unlike these works, in this manuscript we study strategic information transmission in a dynamic timeliness-based gossip network with Markov source evolution. Closest to our setting is \cite{gundogan2025timely}, which considers persuasion
through the timing of truthful disclosures from a binary CTMC source, covering
both the single sender--single receiver case and a single sender broadcasting to
multiple receivers, together with multi-source extensions. In that work the receivers never exchange packets, so there is no network, and each receiver declares the content of its latest received packet. The present model places the declaration rule in a network where packets also arrive through gossip and the receivers choose a network-wide gossip rate. In addition, \cite{tekez2026accuracy} considers a non-strategic sender whose updates are generated at a single rate independent of the source state, whereas we introduce strategic behavior into that gossip system.
\subsection{Contributions}\label{sec:contributions}
\vspace{-0.1cm}

Our contributions are as follows:  $i$) We prove that, at any optimal sender policy with $s^\star\!>\!0$ and $c^\star\!>\!0$, the sender's budget constraint binds. Hence, at a fixed gossip rate, the sender's problem reduces to a one-dimensional optimization over the budget line $s\!+\!c\!=\!R$. $ii$) We prove that in the strategic regime $s \!> \!c$, the receiver utility is strictly increasing and the sender utility is strictly decreasing in the gossip rate $\lambda$, for every $n\!\ge\! 2$. The receivers therefore gossip at the highest available rate. $iii$) Finally, we show that $(c\!-\!s)\tfrac{\partial \US}{\partial\lambda}\!\ge\! 0$ for all parameter regimes. Gossip therefore has no direction of its own, and it moves the network away from the state that the sender pushes harder rather than reinforcing the asymmetry.

A preliminary version of this work~\cite{tekez2026strategic} introduced the model, the utilities, and the Stackelberg formulation, proved gossip monotonicity only for $n=2$, and characterized the equilibrium by invoking a conjecture. Every formal result stated here is either new or strengthens its preliminary counterpart. Four results have no counterpart in the preliminary version: \cref{lem:strict-interior} establishes strict interior bounds on the defect vectors, \cref{thm:budget-binding} proves budget binding for the full two-dimensional sender problem, \cref{cor:direct-sign} settles the sender side in the remaining push-rate regimes, and \cref{cor:cmin-lambda-strategic} separates the direct and policy-adjustment effects of gossip at the equilibrium. For the strengthened results, \cref{thm:UR-monotone} proves gossip monotonicity for every $n\ge 2$, which removes the conjecture from the strategic regime, \cref{cor:full-opt} drops the interval assumption on the feasible set, and \cref{thm:SE} adds uniqueness of the equilibrium together with its closed form. \cref{sec:shs} gives the SHS derivation in full, and \cref{sec:numerics} validates the analytical curves against Monte Carlo simulations.

\section{System Model and Problem Formulation}\label{sec:model}
\vspace{-0.1cm}
In this section, we describe the network model, define the sender and receiver utility functions, and formulate the Stackelberg game between the sender and the receiver (follower) team.
 
\subsection{Network Model}\label{sec:network-model}
\vspace{-0.1cm}

We consider a fully connected (FC) gossip network with $n \ge 2$ nodes, and a sender that observes and disseminates status updates from a time-varying source, as depicted in Fig.~\ref{fig:system_model}. The information at the source evolves according to a binary CTMC with states $Q(t) \in \{0,1\}$ and a generator matrix
\begin{align}
    \mathbf{Q}=\begin{bmatrix} -q_{01} & q_{01} \\[2pt] q_{10} & -q_{10} \end{bmatrix},
\label{eq:q_binary_structure}
\end{align}
where $q_{01}$ is the transition rate from state~$0$ to state~$1$, and $q_{10}$ is the transition rate from state~$1$ to state~$0$. We take $q_{01}>0$ and $q_{10}>0$, so the chain is irreducible. The stationary distribution of the binary CTMC is given by
\begin{align}
    \pi_0=\frac{q_{10}}{q_{01}+q_{10}},
    \;\;
    \pi_1=\frac{q_{01}}{q_{01}+q_{10}}\,\cdot
\label{eq:binary_stationary_pi}
\end{align}
We define $\rho \triangleq q_{01}\pi_0 = q_{10}\pi_1 = \frac{q_{01}q_{10}}{q_{01}+q_{10}}$ as the balanced steady-state flux.

Unlike standard monitoring models where the source update process is exogenous, in this work the sender is \emph{strategic} and controls the rate at which it transmits updates in a \emph{state-dependent} manner. Specifically, the sender transmits version-stamped updates to each node $i \in \mathcal{N}$, where $\mathcal{N} = \{1,\ldots,n\}$ denotes the set of all nodes in the network, according to conditionally independent Poisson processes whose rates depend on the current source state. When $Q(t)=0$, the sender pushes updates to each node at rate $\frac{c}{n}$, and when $Q(t)=1$, it pushes updates to each node at rate $\frac{s}{n}$. The aggregate rate at which the sender transmits is therefore $c$ in state~$0$ and $s$ in state~$1$. Here, $c \ge 0$ and $s \ge 0$ are the sender's design parameters, and we refer to $\theta = (s,c)$ as the \emph{sender policy}. When the state of the CTMC changes, a new version of the update is generated at the source, and we denote the most recent version of the information at the source at time $t$ as $V_0(t)$. Similarly, we denote the most recent version of the update at node $i$ at time $t$ as $V_i(t)$ and its content as $S_i(t) \in \{0,1\}$.
 
When the sender transmits an update to node $i$ at time $t$, node $i$ receives both the current CTMC state and its version information in the form of $\{V_0(t), Q(t)\}$. Since the sender always holds the most recent and accurate information, node $i$ always accepts updates from the sender, \emph{should it decide to opt in}. Thus, upon receiving an update from the sender, node $i$ updates its information as $V_i(t) = V_0(t)$ and $S_i(t) = Q(t)$.
 
In order to measure information freshness at the nodes, we use the version age of information (VAoI) introduced in \cite{yates2021age_gossip, Abolhassani2021}. The version age of node $i$ at time $t$ is denoted by $X_i(t)$ and is given by $X_i(t) \!=\! V_0(t)\!-\!V_i(t)$. In addition to receiving direct updates from the sender, the nodes also share their locally stored packets through \emph{gossiping} in order to improve their information freshness \cite{yates2021age_gossip}. Each ordered pair $(i,j)$ of distinct nodes gossips according to a Poisson process with rate $\frac{\lambda}{n-1}$, where $\lambda \ge 0$ is the per-node gossip rate.\footnote{We consider directed gossip, where only the receiving node's state can change.} When two nodes gossip with one another, they only consider the version age of the incoming packet, as they do not know whether the incoming packet is accurate or not. More formally, when node $i$ sends a packet to node $j$, the version age at node $j$ evolves according to the rule
\begin{align}
X_j(t) = \min\{X_i(t^-), X_j(t^-) \},
\label{eq:x_j_piecewise}
\end{align}
where $t^-$ represents the time instant just prior to the information exchange.

\subsection{Information Accuracy}\label{sec:accuracy}
\vspace{-0.1cm}

Since both sender and receiver utilities will be defined through average node accuracies, we define the \emph{accuracy indicator} for node $i$, denoted $C_i(t)$, as $C_i(t) \!=\! \ind_{\{S_i(t)=Q(t)\}}$, where {$\ind_{\{\cdot\}}$} is the indicator function taking value $1$ only when the event in its argument is true and value $0$ otherwise. We note that the accuracy definition depends only on the content of the source and the relevant node, and does not depend on the version age of the node. Thus, a node can be accurate in two ways, either by holding the most recent source update or by holding a stale packet that is accurate because the source has since returned to the state encoded in that packet.

When node $i$ sends its status update to node $j$, the information accuracy at node $j$ changes according to
\begin{align}
C_j(t) =
\begin{cases}
    C_i(t^-), &\text{if } X_i(t^-) \le X_j(t^-), \\
    C_j(t^-), & \text{if } X_i(t^-) >  X_j(t^-).
\end{cases}
\label{eq:c_j_piecewise}
\end{align}
As a consequence, a node can accept a fresher packet that is actually \emph{less} accurate than its current information, if the source has changed state after the incoming packet was generated.

\subsection{Utility Structure}\label{sec:utility}
\vspace{-0.1cm}

We now define the instantaneous utilities. Each node $j$ declares an estimate $\hat{Q}_j(t)\in\{0,1\}$ of the source state, equal to the content $S_j(t)$ of its most recent packet when it follows the sender. A node therefore declares the content of its freshest stored packet rather than a posterior formed from packet ages, arrival times, and the absence of arrivals; hence, estimators that filter the observed history lie outside the strategy class considered here. Each receiver node $j$ obtains an instantaneous utility given by
\begin{equation}\label{eq:uj-inst}
u_j(t) = q\,\ind_{\{\hat{Q}_j(t)=0,\, Q(t)=0\}} + (1-q)\,\ind_{\{\hat{Q}_j(t)=1,\, Q(t)=1\}}
\end{equation}
for a given importance weight $q \in (0,1)$. 

The sender's instantaneous utility is
\begin{equation}\label{eq:us-inst}
u_s(t) = \frac{1}{n}\sum_{j=1}^{n}\ind_{\{\hat{Q}_j(t)=1\}},
\end{equation}
reflecting the sender's preference for the receivers to hold state-$1$ information regardless of the true source state.

To express the long-run averages of these instantaneous utilities in terms of steady-state quantities, we introduce the mode-tagged mean accuracy components
\begin{equation}\label{eq:cacc-def}
c_{\mathrm{acc}}^{(m)}(\theta,\lambda) = \limt \E\bigl[C_j(t)\,\ind_{\{Q(t)=m\}}\bigr],
\end{equation}
for $m \in \{0,1\}$, where $\theta = (s,c)$. By the symmetry of the fully connected network, all nodes achieve identical steady-state accuracies \cite{tekez2026accuracy}; hence,  $c_{\mathrm{acc}}^{(m)}$ does not depend on the node index $j$.

The receivers' long-run average utility follows directly from \eqref{eq:uj-inst} and \eqref{eq:cacc-def}. Noting that $\ind_{\{\hat{Q}_j=0,\,Q=0\}} = C_j\,\ind_{\{Q=0\}}$ and $\ind_{\{\hat{Q}_j=1,\,Q=1\}} = C_j\,\ind_{\{Q=1\}}$, we obtain
\begin{equation}\label{eq:UR}
U_R(\theta,\lambda) \!=\! \limt \E[u_j(t)] \!=\! q\, c_{\mathrm{acc}}^{(0)}(\theta,\lambda) \!+ \!(1\!-\!q)\, c_{\mathrm{acc}}^{(1)}(\theta,\lambda).\!
\end{equation}

For the sender, we express the long-run average of \eqref{eq:us-inst} by decomposing $\limt \Pr(\hat{Q}_j(t)=1)$ over the two source modes. Conditioning on the source mode gives $\Pr(\hat{Q}_j\!=\!1 \mid Q\!=\!0) = 1 - \E[C_j \mid Q\!=\!0]$ and $\Pr(\hat{Q}_j\!=\!1 \mid Q\!=\!1) = \E[C_j \mid Q\!=\!1]$. Weighting these by $\pi_0$ and $\pi_1$ and using $c_{\mathrm{acc}}^{(m)} = \pi_m\,\limt \E[C_j(t) \mid Q(t)=m]$ gives $\limt \Pr(\hat{Q}_j(t)=1) = \pi_0 - c_{\mathrm{acc}}^{(0)} + c_{\mathrm{acc}}^{(1)}$. Averaging over the nodes, the sender's long-run average utility is
\begin{equation}\label{eq:US}
U_S(\theta,\lambda) = \limt \E[u_s(t)] = \pi_0 - c_{\mathrm{acc}}^{(0)}(\theta,\lambda) + c_{\mathrm{acc}}^{(1)}(\theta,\lambda).
\end{equation}

\subsection{Stackelberg Game Formulation}\label{sec:stackelberg}
\vspace{-0.1cm}

We first state a standing assumption that settles the receiver's default strategy.

\begin{assumption}
The parameters satisfy $q\pi_0 > (1-q)\pi_1$, or equivalently $\Delta := q\,q_{10} - (1-q)\,q_{01} > 0$.\label{as:outside}
\end{assumption}

Under this assumption, the receiver's optimal constant estimate in the absence of any sender information is $\hat{Q}_j(t) = 0$, yielding outside-option utility $U_{\mathrm{out}} = q\pi_0$, and without updates the receivers never declare state~$1$, which gives the sender zero utility. Assumption~\ref{as:outside} restricts attention to the nontrivial case in which the receivers prefer the constant estimate state~$0$, whereas the sender prefers state~$1$. If the receivers instead preferred the constant estimate state~$1$, the default strategy would already give the sender its maximum utility, leaving no scope for strategic persuasion. Following the sender requires the receivers to operate the gossip protocol and to process the incoming packets, which costs them a fixed amount $\eta>0$ per unit time. The receiver team therefore collects $\UR(\theta,\lambda)\!-\!\eta$ when it follows and $q\pi_0$ when it does not, so the participation constraint (PC) requires $\UR(\theta,\lambda) \ge q\pi_0 + \eta$. Since $c_{\mathrm{acc}}^{(m)}\!\le\!\pi_m$ by~\eqref{eq:cacc-def}, we have $\UR\!\le\! q\pi_0\!+\!(1\!-\!q)\pi_1$, so no policy meets the PC unless $\eta\le(1-q)\pi_1$. \cref{lem:strict-interior} will place $\UR$ strictly below that bound at every finite policy, so the PC requires $\eta<(1-q)\pi_1$. The sender must therefore provide sufficient information quality to ensure receivers' participation, while limiting state-$0$ updates that act against its interest.

The interaction is formulated as a Stackelberg game, in which the sender acts as the leader and the receiver team as the follower. The sender first commits to a policy $\theta = (s,c) \in \Theta = \{(s,c) \in \reals_+^2 : 0 < s + c \leq R\}$, where $R$ represents the sender budget.\footnote{The origin is excluded because no source updates arrive when $s=c=0$, and the steady-state packet contents depend on the initial contents. The receivers then use the default estimate $\hat{Q}_j(t)=0$, which gives $U_R=q\pi_0$ and $U_S=0$.} The average rate at which the sender transmits under this budget is $\pi_0 c+\pi_1 s$ by~\eqref{eq:binary_stationary_pi}. Then, the receiver team observes the sender's committed policy and selects its best response. At this point, the receiver team has two options. If following the sender at some gossip rate $\lambda \in [0,\Lambda]$ meets the PC, the receiver team adopts $\sigma_{\mathrm{follow}}(\lambda)$, under which it follows the sender's messages and gossips at rate $\lambda$, and the sender collects $\US(\theta,\lambda)$ of~\eqref{eq:US}. Otherwise it adopts $\sigma_{\mathrm{default}}$, under which it disregards the sender's messages and holds its estimate at state~$0$ at all times, obtaining $U_{\mathrm{out}} = q\pi_0$ while~\eqref{eq:us-inst} gives $\US(\theta,\sigma_{\mathrm{default}})=0$.

All receiver nodes attain the same utility $\UR(\theta,\lambda)$, and $\lambda$ is common to the network; hence, maximizing $\UR$ maximizes the objective of the receiver team. Let $U_R^\star(\theta)\triangleq\max_{\lambda\in[0,\Lambda]}\UR(\theta,\lambda)$, and let $\lambda^\star(\theta)$ be a receiver-optimal gossip rate. We define the optimistic Stackelberg equilibrium as the sender policy $\theta^\star$ satisfying
\begin{align} \label{eq:st_formulation}
 \US(\theta^\star, \mathrm{BR}(\theta^\star)) \ge \US(\theta, \mathrm{BR}(\theta)),
\end{align}
for all $\theta \in \Theta$, where the best response of the receiver team is given by
\begin{equation}\label{eq:BR}
\mathrm{BR}(\theta) = \begin{cases}
\sigma_{\mathrm{follow}}(\lambda^\star(\theta)), & \text{if } U_R^\star(\theta) \ge q\pi_0 + \eta, \\
\sigma_{\mathrm{default}}, & \text{otherwise}.
\end{cases}
\end{equation}
The equilibrium is optimistic because two ties are resolved in favor of the sender, the one at $U_R^\star(\theta)\!=\!q\pi_0\!+\!\eta$ by following in~\eqref{eq:BR}, and the one among the maximizers of $\UR(\theta,\cdot)$ by the choice of $\lambda^\star(\theta)$.
\section{The SHS Model and the Accuracy Recursion}\label{sec:shs}
\vspace{-0.1cm}

Both utilities in~\eqref{eq:UR}--\eqref{eq:US} depend on the steady-state mode-tagged accuracies $c_{\mathrm{acc}}^{(0)}$ and $c_{\mathrm{acc}}^{(1)}$, which are shaped jointly by the source transitions, the sender's state-dependent pushes, and the gossip process, all operating concurrently in continuous time. We compute them with the stochastic hybrid systems (SHS) framework~\cite{hespanha2006modelling}, which derives moment equations for piecewise-deterministic Markov processes by applying a generator identity to suitably chosen test functions. The construction introduces the \emph{freshest-node accuracy} vector $\mathbf{f}_k$, which tracks the accuracy of the freshest node within any $k$-node subset, tagged by the source state.

\subsection{The SHS Model}
\vspace{-0.1cm}

We follow the accuracy recursion framework in~\cite{tekez2026accuracy}, with the source-push rate replaced by the state-dependent sender policy.

Similar to \cite{tekez2026accuracy}, we define the hybrid state as
\begin{equation}
\bigl(Q(t),\mathbf Z(t)\bigr)
= \bigl(Q(t),\mathbf C(t),\mathbf X(t)\bigr),
\label{eq:hybrid_state}
\end{equation}
where $Q(t)\in\{0,1\}$ is the discrete source mode, $\mathbf C(t)=[C_1(t),\ldots,C_n(t)]\in\{0,1\}^n$ is the accuracy vector, and $\mathbf X(t)=[X_1(t),\ldots,X_n(t)]\in \mathbb Z_{\ge 0}^n$ is the version-age vector. Between jumps, $\mathbf Z(t)$ is constant, i.e., $\dot{\mathbf{Z}}(t)=0$. With $\bar m=1-m$, the set of all transitions for this SHS is given by
\begin{align}
\mathcal E
\!&=\! \{(0,0,m\to\bar m)\!\mid\! m\!\in\!\{0,1\}\}\cup \{(i,j)\mid i,j\!\in\!\mathcal N,\! i\neq j\}\nonumber \\
&\;\cup \{(0,m\to j)\mid m\in\{0,1\},\ j\in\mathcal N\}.
\label{eq:set_of_transitions}
\end{align}

The transition $(0,0,m\to\bar m)$ is a source-mode change (a source CTMC state flip), $(0,m\to j)$ is a source push to node $j$ while the source CTMC is in mode $m$, and $(i,j)$ denotes a directed gossip event from node $i$ to node $j$. Each transition $e\in\mathcal E$ resets the hybrid state through the reset map
\begin{equation}
\bigl(Q(t^+),\mathbf Z(t^+)\bigr)
= \phi_e\bigl(Q(t^-),\mathbf Z(t^-)\bigr).
\label{eq:q_z_reset}
\end{equation}
The transition rates for the events defined in~\eqref{eq:set_of_transitions} are
\begin{equation}
\lambda_e(Q)
=
\begin{cases}
q_{m\bar m}\,\mathds{1}_{\{Q=m\}},
& e=(0,0,m\to\bar m),\\[3pt]
\dfrac{c}{n}\,\mathds{1}_{\{Q=0\}},
& e=(0,0\to j),\\[7pt]
\dfrac{s}{n}\,\mathds{1}_{\{Q=1\}},
& e=(0,1\to j),\\[7pt]
\dfrac{\lambda}{n-1},
& e=(i,j).
\end{cases}
\label{eq:lambda_e_def}
\end{equation}
Thus, the sender policy enters the SHS only through the mode-dependent source-push rates, while gossip remains mode independent.

\subsection{Reset Maps}\label{sec:reset_maps}
\vspace{-0.1cm}

The node-level update rules in~\eqref{eq:x_j_piecewise} and~\eqref{eq:c_j_piecewise} induce the set-level reset maps needed for the SHS recursion. For a subset $A_k\subseteq\mathcal N$ with cardinality $|A_k|=k$, we define the \emph{version age of set} $A_k$ as the version age of its freshest node,
\begin{align}
    X_{A_k}(t) \triangleq \min_{j \in A_k} X_j(t),
    \label{eq:version_age_subset}
\end{align}
which determines the minimum version age of the packet that the set $A_k$ holds at time $t$. We define the freshness-based accuracy of $A_k$ as the accuracy indicator of the freshest node of $A_k$,
\begin{equation}
F_{A_k}(\mathbf C,\mathbf X)
\triangleq
C_{\argmin_{\ell\in A_k} X_\ell},
\label{eq:f_ak_defn}
\end{equation}
with arbitrary tie breaking. The version counter increases at every source transition, so two nodes holding the same version hold the same content and the same accuracy indicator. The tie breaking therefore leaves~\eqref{eq:f_ak_defn} unchanged. We will use this test function of the hybrid state to characterize the average accuracy terms that make up the utilities defined in the earlier section. The reset map for the set-level version age is given by
\begin{align}
X'_{A_k}
&=
\begin{cases}
X_{A_k}+1,
& e=(0,0,m\to\bar m),\\
0,
& e=(0,m\to j),\ j\in A_k,\\
X_{A_{k+1}},
& e=(i,j),\ i\in\mathcal N\setminus A_k,\ j\in A_k,\\
X_{A_k},
& \text{otherwise.}
\end{cases}
\label{eq:X_Ak_update}
\end{align}
Similarly, we provide the reset map for the freshness-based accuracy as
\begin{align}
F'_{A_k}
&=
\begin{cases}
1-F_{A_k},
& e=(0,0,m\to\bar m),\\
1,
& e=(0,m\to j),\ j\in A_k,\\
F_{A_{k+1}},
& e=(i,j),\ i\in\mathcal N\setminus A_k,\ j\in A_k,\\
F_{A_k},
& \text{otherwise.}
\end{cases}
\label{eq:F_Ak_update}
\end{align}
A source-mode change increments the version age of every node, and therefore of every subset, by one, and it flips the binary accuracy indicators. A source push to a node in $A_k$ resets that node’s version age to zero and makes its stored packet accurate. An inbound gossip event can change the freshest packet of the set, so the expanded set $A_k\cup\{i\}$ enters, which network symmetry lets us write as $A_{k+1}$. See~\cite{tekez2026accuracy} for a detailed discussion of these reset maps.

We apply the steady-state SHS balance identity to these set quantities to derive the steady-state values for the mode-tagged freshest-node accuracies. For any time-invariant test function $\psi$ of the hybrid state, we obtain from stationarity
\begin{equation}
\sum_{e\in\mathcal E}
\mathbb E\!\left[
\lambda_e(Q)
\left(
\psi\bigl(\phi_e(Q,\mathbf Z)\bigr)
- \psi(Q,\mathbf Z)
\right)
\right]
=0.
\label{eq:shs_balance_equation}
\end{equation}
We use~\eqref{eq:shs_balance_equation} with mode-tagged versions of $F_{A_k}$ that will correspond to the time invariant test functions mentioned above. This yields a backward recursion for the freshest-node accuracy, which we provide in the next subsection.

\subsection{Accuracy Recursion}\label{sec:fk_recursion}
\vspace{-0.1cm}
In this subsection, we derive a backward recursion for the freshest-node accuracy and connect it to the per-node accuracies that appear in the utilities. To this end, we define the mode-tagged freshest-accuracy vector as
\begin{equation}
\mathbf f_k
=
\left[
\begin{smallmatrix}
f_k^{(0)}\\
f_k^{(1)}
\end{smallmatrix}
\right],
\;\;
f_k^{(m)}
=
\limt
\mathbb E\!\left[
F_{A_k}(t)\,\mathds{1}_{\{Q(t)=m\}}
\right].
\label{eq:fk_vec_def}
\end{equation}
For a single-node set $A_1=\{j\}$ the freshest node is node $j$ itself, so $F_{A_1}(t)=C_j(t)$. Hence, comparing~\eqref{eq:fk_vec_def} with~\eqref{eq:cacc-def} gives $f_1^{(m)}\!=\!c_{\mathrm{acc}}^{(m)}$ for $m\!\in\!\{0,1\}$, so the per-node average accuracies are equal to the $k=1$ member of the $f_k$ family. For $k>1$ the vector $\mathbf f_k$ tracks the freshest node of a $k$-set rather than a fixed node, so the identity holds only at $k=1$ \cite[Remark~1]{tekez2026accuracy}. The utilities in~\eqref{eq:UR} and~\eqref{eq:US} therefore can be written as
\begin{equation}\label{eq:util-f1}
\UR=q f_1^{(0)}+(1-q) f_1^{(1)},\;\;
\US=\pi_0-f_1^{(0)}+f_1^{(1)},
\end{equation}
and we work with $\mathbf f_1$ from here on.

For $k=1,\ldots,n$, we define
\begin{equation}
\alpha_k^{(0)}
\triangleq
\frac{kc}{n},
\;\;
\alpha_k^{(1)}
\triangleq
\frac{ks}{n},
\;\;
\alpha_k
\triangleq
\frac{k(n-k)\lambda}{n-1}.
\label{eq:alpha_def}
\end{equation}
The first two terms denote the total source-push rates into a set $A_k$ in modes $0$ and $1$, while $\alpha_k$ is the total rate of gossip events from outsider nodes, $i \in \mathcal N\setminus A_k$, into insider nodes, $j \in A_k$. Then, we define
\begin{equation}
\!\!\mathbf{W}_k
\!\triangleq\!
\left[
\begin{smallmatrix}
q_{01}+\alpha_k^{(0)}+\alpha_k & q_{10}\\
q_{01} & q_{10}+\alpha_k^{(1)}+\alpha_k
\end{smallmatrix}
\right],
\!\!\;
\mathbf v_k\!\!
\triangleq\!\!
\left[
\begin{smallmatrix}
q_{10}\pi_1+\alpha_k^{(0)}\pi_0\\
q_{01}\pi_0+\alpha_k^{(1)}\pi_1
\end{smallmatrix}
\right].\!\!
\label{eq:Wk_vk_def}
\end{equation}
We now state the finite-dimensional SHS recursion used in the remainder of the manuscript.

\begin{theorem}\label{thm:backward_recursion}
Let $n\ge 2$, $q_{01}>0$, $q_{10}>0$, and $s+c>0$. For $k=1,\ldots,n-1$, the mode-tagged freshest-accuracy vectors satisfy
\begin{equation}
\mathbf{W}_k\mathbf f_k
=
\mathbf v_k+\alpha_k\mathbf f_{k+1},
\label{eq:fk_recursion}
\end{equation}
with the boundary condition given by $\mathbf{W}_n\mathbf f_n=\mathbf v_n.$ Here, $\alpha_n^{(0)}=c$, $\alpha_n^{(1)}=s$, and $\alpha_n=0$, which removes the level-$(n+1)$ term from this recursion and from every recursion derived from it, so no quantity above level $n$ is used. Writing $D_n=q_{01}s+cq_{10}+cs$, we equivalently have
\begin{equation}
\mathbf f_n
=
\left[
\begin{smallmatrix}
\dfrac{c\,\pi_0(q_{10}+s)}{D_n}\\[5pt]
\dfrac{s\,\pi_1(q_{01}+c)}{D_n}
\end{smallmatrix}
\right].
\label{eq:fn_closed}
\end{equation}
Each $\mathbf f_k$ is uniquely determined by backward induction in~\eqref{eq:fk_recursion}.
\end{theorem}

\begin{proof}
We apply the balance identity~\eqref{eq:shs_balance_equation} to $F_{A_k}\mathds{1}_{\{Q=0\}}$ and $F_{A_k}\mathds{1}_{\{Q=1\}}$, alongside the reset map~\eqref{eq:F_Ak_update}. This yields the following two balance equations:
\begin{align}
 \!q_{10}\!\left(\!\pi_1\!-\!\!f_k^{(1)}\!\right)\! -\! q_{01}f_k^{(0)}\! \!+ \!\!\alpha_k^{(0)}\!\left(\!\pi_0\!-\!f_k^{(0)}\!\right)\! + \!\alpha_k\left(\!f_{k+1}^{(0)}\!-\!f_k^{(0)}\!\right)\!\!=& 0,
\label{eq:F_balance0}\\
 \!q_{01}\!\left(\!\pi_0\!-\!f_k^{(0)}\!\right) \!\!- \!q_{10}f_k^{(1)} \!\!+\! \alpha_k^{(1)}\!\left(\!\pi_1\!-\!f_k^{(1)}\!\right) \!+ \!\alpha_k\!\left(\!f_{k+1}^{(1)}\!-\!f_k^{(1)}\!\right)\!\!=& 0.
\label{eq:F_balance1}
\end{align}
Rearranging~\eqref{eq:F_balance0} and~\eqref{eq:F_balance1} gives~\eqref{eq:fk_recursion}. When $k=n$, there are no nodes outside $A_n$, so $\alpha_n=0$. Solving the resulting two-dimensional linear system gives~\eqref{eq:fn_closed}, where we use the fact $q_{01}\pi_0=q_{10}\pi_1$ for simplification.

Expanding the determinant of $\mathbf{W}_k$ in~\eqref{eq:Wk_vk_def} gives $\!\det \mathbf{W}_k \!=\! q_{01}(\alpha_k^{(1)}\!+\!\alpha_k) \!+\! q_{10}(\alpha_k^{(0)}\!+\!\alpha_k) \!+\! (\alpha_k^{(0)}\!+\!\alpha_k)(\alpha_k^{(1)}\!\!+\!\alpha_k)$. Since $\alpha_k^{(0)}\!+\!\alpha_k^{(1)}\!>\!0$ for every $k$ when $c\!+\!s\!>\!0$, we have $\!\det \!\mathbf{W}_k\!>\!0$ whenever $q_{01},q_{10}\!>\!0$ and $c\!+\!s\!>\!0$. Hence each $\mathbf{W}_k$ is nonsingular, and the recursion determines $\!\mathbf f_k\!$ uniquely from $\mathbf f_{k+1}$.
\end{proof}

Since $F_{A_k}\in\{0,1\}$, the entry $f_k^{(m)}$ in~\eqref{eq:fk_vec_def} is the joint probability that the source is in mode $m$ and the freshest packet of $A_k$ is accurate. Its complement, which we call the \emph{defect}, appears throughout the rest of the manuscript, with
\begin{equation}\label{eq:defect-def}
\hvec_k=\bmat{h_k^{(0)}\\ h_k^{(1)}},\;
h_k^{(m)}\!=\!\pi_m\!-\!f_k^{(m)}\!=\!\Pr\bigl(\!F_{A_k}\!\!=\!0,Q\!=\!m\bigr).
\end{equation}
The following lemma keeps every component of $\fvec_k$ and $\hvec_k$ away from both of its extremes when $s,c>0$.

\begin{lemma}\label{lem:strict-interior}
Let $n\ge 2$, $q_{01}>0$, $q_{10}>0$, $s,c>0$, and $\lambda\ge 0$. For every $k=1,\ldots,n$, we have
\begin{equation}\label{eq:hk-strict}
\mathbf 0<\fvec_k<\bmat{\pi_0\\ \pi_1},\;\;\text{equivalently}\;\;
\mathbf 0<\hvec_k<\bmat{\pi_0\\ \pi_1},
\end{equation}
where the inequalities are componentwise.
\end{lemma}

\begin{proof}
Since $F_{A_k}\in\{0,1\}$, the definition in~\eqref{eq:fk_vec_def} gives $f_k^{(m)} =\Pr\bigl(F_{A_k}=1,Q=m\bigr),$ and hence $0\le f_k^{(m)}\le \Pr(Q=m)=\pi_m.$ It therefore remains to rule out equality at either endpoint. At the boundary level $k=n$, both numerators and the denominator $D_n=q_{01}s+cq_{10}+cs$ in~\eqref{eq:fn_closed} are strictly positive. Thus,
$\fvec_n>\mathbf 0$. Moreover, we have $D_n-c(q_{10}+s)=q_{01}s>0$ and $D_n-s(q_{01}+c)=cq_{10}>0.$ It follows from~\eqref{eq:fn_closed} that
$f_n^{(0)}<\pi_0$ and $f_n^{(1)}<\pi_1$.

Now fix $k\in\{1,\ldots,n-1\}$. Writing the two rows of
\eqref{eq:fk_recursion}, with $\mathbf W_k$ and $\mathbf v_k$ defined
in~\eqref{eq:Wk_vk_def}, gives
\begin{align}
\bigl(\!q_{01}\!+\!\alpha_k^{(0)}\!\!+\!\alpha_k\!\bigr)f_k^{(0)}\!\!+\!q_{10}f_k^{(1)}\!\!&=\!q_{10}\pi_1\!+\!\alpha_k^{(0)}\pi_0\!+\!\alpha_k f_{k+1}^{(0)},\!\!\label{eq:fk-row0}\\
q_{01}f_k^{(0)}\!\!+\!\bigl(\!q_{10}\!+\!\alpha_k^{(1)}\!\!+\!\alpha_k\!\bigr)f_k^{(1)}\!\!&=\!q_{01}\pi_0\!+\!\alpha_k^{(1)}\pi_1\!+\!\alpha_k f_{k+1}^{(1)},\!\!\label{eq:fk-row1}
\end{align}
The probabilistic definition of $f_{k+1}^{(m)}$ gives
$0\le f_{k+1}^{(m)}\le\pi_m$.

We first establish strict positivity. Suppose that $f_k^{(0)}=0$.
Then~\eqref{eq:fk-row0} gives $q_{10}f_k^{(1)}
=
q_{10}\pi_1+\alpha_k^{(0)}\pi_0
+\alpha_k f_{k+1}^{(0)}
>
q_{10}\pi_1,$ where the strict inequality follows from
$\alpha_k^{(0)}=\frac{kc}{n}>0$. Hence,
$f_k^{(1)}>\pi_1$, contradicting the upper bound
$f_k^{(1)}\le\pi_1$. Similarly, suppose that $f_k^{(1)}=0$. Then~\eqref{eq:fk-row1} gives $q_{01}f_k^{(0)}
=q_{01}\pi_0+\alpha_k^{(1)}\pi_1
+\alpha_k f_{k+1}^{(1)}>q_{01}\pi_0,$ because $\alpha_k^{(1)}=\frac{ks}{n}>0$. Hence,
$f_k^{(0)}>\pi_0$, contradicting
$f_k^{(0)}\le\pi_0$. Therefore, we have
$\fvec_k>\mathbf 0$.

We next establish the strict upper bounds. Suppose that
$f_k^{(0)}=\pi_0$. Substituting this equality into
\eqref{eq:fk-row0} and using
$q_{01}\pi_0=q_{10}\pi_1$ yields $q_{10}f_k^{(1)}=\alpha_k\bigl(f_{k+1}^{(0)}-\pi_0\bigr)=-\alpha_k h_{k+1}^{(0)}\le 0,$
where the inequality follows from $\alpha_k\ge0$ and
$h_{k+1}^{(0)}\ge0$. This contradicts the already established
inequality $f_k^{(1)}>0$. Similarly, suppose that $f_k^{(1)}=\pi_1$. Substituting this equality
into~\eqref{eq:fk-row1} and again using
$q_{01}\pi_0=q_{10}\pi_1$ gives $q_{01}f_k^{(0)}
=
\alpha_k\bigl(f_{k+1}^{(1)}-\pi_1\bigr)
=
-\alpha_k h_{k+1}^{(1)}
\le 0,$ which contradicts $f_k^{(0)}>0$. Therefore, we have $\mathbf 0<\fvec_k<\bmat{\pi_0\\ \pi_1}$ for every $k=1,\ldots,n$. Finally, since $\hvec_k=\bmat{\pi_0\\ \pi_1}-\fvec_k$ by~\eqref{eq:defect-def}, the preceding inequalities are equivalent to $\mathbf 0<\hvec_k<\bmat{\pi_0\\ \pi_1}.$
\end{proof}

Intuitively, the packet that is freshest in $A_k$ records the source state at its generation time. It is accurate if the source has undergone an even number of state transitions since then and inaccurate if it has undergone an odd number. Because $q_{01},q_{10},s,$ and $c$ are strictly positive, both types of histories occur with positive steady-state probability for each current source mode. Equation~\eqref{eq:hk-strict} formalizes this strict nondegeneracy. The steady state therefore follows from a single backward pass. We compute $\mathbf f_n$ from~\eqref{eq:fn_closed} and apply~\eqref{eq:fk_recursion} down to $k=1$, where the subset is a single node and~\eqref{eq:alpha_def} gives $\alpha_1^{(0)}=\tfrac{c}{n}$, $\alpha_1^{(1)}=\tfrac{s}{n}$, and $\alpha_1=\lambda$. Substituting $\mathbf f_1$ into~\eqref{eq:util-f1} yields the receiver and sender utilities for any sender policy $(s,c)$ and gossip rate $\lambda$.

Evaluating both utilities at a given $(\theta,\lambda)$ costs $n$ solves of a $2\times 2$ system, so the whole steady state is obtained in $\mathcal{O}(n)$ operations, even though the accuracy vector $\mathbf C(t)$ ranges over $2^n$ configurations. This keeps each utility evaluation inexpensive in the scalar search developed in \cref{sec:equilibrium}.

\subsection{Sign-Preservation Property}\label{sec:sign-lemma}
\vspace{-0.1cm}

The $2 \times 2$ matrices $\mathbf{W}_k$ in the backward
recursion~\eqref{eq:fk_recursion} share a structure that preserves
component-wise sign patterns under inversion.

\begin{lemma}\label{lem:sign}
Let $\mathbf{M} = \bigl[\begin{smallmatrix} m_{11} & m_{12} \\ m_{21} & m_{22}
\end{smallmatrix}\bigr]$ with $m_{11},m_{12},m_{21},m_{22} > 0$ and $\det(\mathbf{M}) > 0$.
If\, $\mathbf{y} = [y_1\;\; y_2]^\top$ satisfies $y_1 \ge 0$ and
$y_2 \le 0$, then $\mathbf{x} = \mathbf{M}^{-1}\mathbf{y}$ satisfies
$x_1 \ge 0$ and $x_2 \le 0$. Similarly, if $\mathbf{y} = [y_1\;\; y_2]^\top$ satisfies $y_1\leq 0$ and $y_2\geq 0$, then $x_1 \le 0$ and $x_2 \ge 0$ are satisfied under the same assumptions.
\end{lemma}
\begin{proof}
Since $\det(\mathbf{M}) > 0$, the components of $\mathbf{x} = \mathbf{M}^{-1}\mathbf{y}$ are
$x_1 = \tfrac{m_{22}\,y_1 - m_{12}\,y_2}{\det(\mathbf{M})}$ and
$x_2 = \tfrac{m_{11}\,y_2 - m_{21}\,y_1}{\det(\mathbf{M})}$.
When $y_1 \ge 0$ and $y_2 \le 0$, the numerator of $x_1$ is non-negative and the numerator of $x_2$ is non-positive. The reversed-sign case follows by applying the same argument to $-\mathbf{y}$.
\end{proof}
Every $\mathbf{W}_k$ in~\eqref{eq:Wk_vk_def} has positive entries and, by the proof of \cref{thm:backward_recursion}, positive determinant, so \cref{lem:sign} applies to $\mathbf{W}_k^{-1}$ at every level and drives the inductions of \cref{sec:sender-opt,sec:monotone}. We next use these closed-form expressions to optimize the sender’s rate allocation between the two source states.
\section{Sender Optimality and Budget Binding}\label{sec:sender-opt}
\vspace{-0.1cm}

In this section, we first show that the sender utility is monotone in the push rates, and then prove that the budget constraint binds at any interior optimum.

\subsection{Sensitivity Analysis and Monotonicity of the Sender Utility}\label{sec:sender-mono}
\vspace{-0.1cm}

We begin by analyzing the \emph{sensitivity} of the recursion to the sender push rates $s$ and $c$. Let $\mathbf e_1=[1\;\;0]^\top$ and
$\mathbf e_2=[0\;\;1]^\top$. Differentiating the mode-tagged freshest-accuracy vectors $\mathbf f_k$ given in \eqref{eq:fk_recursion} with respect to $c$ and to $s$, denoted by $\mathbf f_{k,c}$ and $\mathbf f_{k,s}$, respectively, gives
\begin{align}
\mathbf{W}_k\,\fvec_{k,c} &= \alpha_k\,\fvec_{k+1,c} + \tfrac{k}{n}h_k^{(0)}\,\eone, \label{eq:fk-c-sens}\\
\mathbf{W}_k\,\fvec_{k,s} &= \alpha_k\,\fvec_{k+1,s} + \tfrac{k}{n}h_k^{(1)}\,\etwo, \label{eq:fk-s-sens}
\end{align}
for $k=1,\ldots,n-1$, with the boundary equations $\mathbf{W}_n\fvec_{n,c}=h_n^{(0)}\eone$ and $\mathbf{W}_n\fvec_{n,s}=h_n^{(1)}\etwo$. In terms of~\eqref{eq:defect-def}, the forcing terms of the two recursions are $\tfrac{k}{n}h_k^{(0)}\eone$ and $\tfrac{k}{n}h_k^{(1)}\etwo$, with sign patterns $(+,0)$ and $(0,+)$ by~\eqref{eq:hk-strict}. The remaining term carries the sign pattern of level $k+1$, which is absent at $k=n$. Starting there and working down, \cref{lem:sign} gives the following signs by backward induction.

\begin{prop}\label{prop:signs-mono}
For $s, c > 0$, $\lambda \ge 0$, and every $k=1,\ldots,n$,
\begin{equation}\label{eq:acc-signs}
\frac{\partial f_k^{(0)}}{\partial c} > 0, \;
\frac{\partial f_k^{(1)}}{\partial c} < 0, \;
\frac{\partial f_k^{(0)}}{\partial s} < 0, \;
\frac{\partial f_k^{(1)}}{\partial s} > 0.
\end{equation}
Consequently, we have
\begin{equation}\label{eq:US-mono}
\frac{\partial \US}{\partial c} < 0, \;\; \frac{\partial \US}{\partial s} > 0.
\end{equation}
\end{prop}

\begin{proof}
At the boundary level $k=n$ we have $\alpha_n=0$, so~\eqref{eq:fk-c-sens} reduces to $\mathbf{W}_n\fvec_{n,c}=h_n^{(0)}\eone$, whose right-hand side is nonzero by~\eqref{eq:hk-strict}. Since $\mathbf{W}_n^{-1}$ has positive diagonal entries and negative off-diagonal entries, the boundary equation implies that $f_{n,c}^{(0)}>0$ and $f_{n,c}^{(1)}<0$. For the induction step, suppose that $\fvec_{k+1,c}$ has the sign pattern $(+,-)$. If $\lambda>0$, then $\alpha_k\fvec_{k+1,c}$ has the same sign pattern, whereas this term is zero if $\lambda=0$. The remaining term $\tfrac{k}{n}h_k^{(0)}\eone$ has the sign pattern $(+,0)$. Therefore, the right-hand side of~\eqref{eq:fk-c-sens} has a strictly positive first component and a nonpositive second component.
Lemma~\ref{lem:sign} then preserves the pattern, and backward induction to $k=1$ gives the first two inequalities of \eqref{eq:acc-signs}. The $s$-sensitivity argument is symmetric, with the forcing along $\etwo$ producing the sign pattern $(-,+)$.

By \eqref{eq:util-f1}, we have $\US=\pi_0-f_1^{(0)}+f_1^{(1)}$. Therefore, applying the sign relations in \eqref{eq:acc-signs} with $k=1$ yields
\begin{equation*}
\frac{\partial \US}{\partial c}
=-f_{1,c}^{(0)}+f_{1,c}^{(1)}<0,
\;\;
\frac{\partial \US}{\partial s}
=-f_{1,s}^{(0)}+f_{1,s}^{(1)}>0,
\end{equation*}
establishing the monotonicity properties stated in \eqref{eq:US-mono}.
\end{proof}

Intuitively, increasing $c$ injects more state-$0$ packets, which raises the mode-$0$ accuracy. Those same packets keep circulating after the source has moved to state~$1$, which lowers the mode-$1$ accuracy. Increasing $s$ has the mirror effect.

\subsection{Budget Binding}\label{sec:budget-binding}
\vspace{-0.1cm}

In Proposition~\ref{prop:signs-mono}, we show that $\frac{\partial \US}{\partial s} {}> 0$, suggesting that the sender should exhaust its budget. However, because $\frac{\partial \UR}{\partial s} $ does not have a fixed sign, increasing $s$ alone may violate the PC. A two-dimensional argument is therefore needed to show that the budget constraint binds at every interior optimum. We introduce the following compact notation for the accuracy sensitivities,
\begin{equation}\label{eq:abpr}
f_{1,c}^{(0)}\! :=\! \frac{\partial f_1^{(0)}}{\partial c}, 
 \bar{f}_{1,c}^{(1)} \!:= \!-\frac{\partial f_1^{(1)}}{\partial c}, 
\bar{f}_{1,s}^{(0)} \! :=\! -\frac{\partial f_1^{(0)}}{\partial s}, 
f_{1,s}^{(1)}\! :=\! \frac{\partial f_1^{(1)}}{\partial s},
\end{equation}
all of which are strictly positive by Proposition~\ref{prop:signs-mono}. Writing the gradients in the $(s,c)$ coordinate order, the utility gradients take the form
\begin{align}
\nabla \US \!&=\! \left(\frac{\partial \US}{\partial s},\frac{\partial \US}{\partial c} \right) = \bigl(\bar{f}_{1,s}^{(0)}\!+\!f_{1,s}^{(1)},\; -( f_{1,c}^{(0)}\!+\!\bar{f}_{1,c}^{(1)})\bigr),\!\! \label{eq:grad-US}\\
\nabla \UR \!&=\!\left(\frac{\partial \UR}{\partial s},\frac{\partial \UR}{\partial c} \right)\!\! \nonumber \\ &=\!\! \bigl(-q\bar{f}_{1,s}^{(0)}\!+\!(1\!-\!q)f_{1,s}^{(1)},\; q f_{1,c}^{(0)}\!-\!(1\!-\!q)\bar{f}_{1,c}^{(1)}\bigr).\!\! \label{eq:grad-UR}
\end{align}

The argument relies on the positivity of the determinant of the accuracy-sensitivity matrix $\mathbf{J}_k \triangleq \bigl[\,\fvec_{k,c}\ \ \fvec_{k,s}\,\bigr]$, whose columns collect the derivatives of $\mathbf f_k$ with respect to the two push rates. At $k=1$, the sensitivities in~\eqref{eq:abpr} identify its columns as $\fvec_{1,c} = [\,f_{1,c}^{(0)}\ \ {-\bar{f}_{1,c}^{(1)}}\,]^\top$ and $\fvec_{1,s} = [\,{-\bar{f}_{1,s}^{(0)}}\ \ f_{1,s}^{(1)}\,]^\top$, so $\det \mathbf{J}_1 = f_{1,c}^{(0)} f_{1,s}^{(1)} - \bar{f}_{1,c}^{(1)}\,\bar{f}_{1,s}^{(0)}$. A positive determinant means the two push rates change the accuracy vector in different directions, which is what lets the sender trade one against the other. The recursion forces this determinant to be positive, which is established as the first step of the proof below.

\begin{theorem}\label{thm:budget-binding} Let $\lambda \!\ge\! 0$, $R\! >\! 0$, and $\eta \!> \!0$ be fixed. Let $(s^\star, c^\star)$ be an optimal PC-feasible sender policy at this fixed gossip rate $\lambda$, and suppose that $s^\star > 0$ and $c^\star > 0$. Then, we have $s^\star + c^\star = R$. \end{theorem}
\begin{proof}
\emph{Step 1 (Determinant positivity).} We first show that $\det \mathbf{J}_k>0$ for every $k = 1,\ldots,n$ and every $s,c>0$ and $\lambda\ge0$, where
\begin{equation}\label{eq:det-fk-pos}
\det \mathbf{J}_k
= f_{k,c}^{(0)} f_{k,s}^{(1)} - f_{k,c}^{(1)} f_{k,s}^{(0)} .
\end{equation}
The forcing magnitudes of \eqref{eq:fk-c-sens} and \eqref{eq:fk-s-sens} are $\tfrac{k}{n}h_k^{(0)}$ and $\tfrac{k}{n}h_k^{(1)}$, both positive by~\eqref{eq:hk-strict}, and \cref{prop:signs-mono} gives $f_{k+1,c}^{(0)}>0$ and $f_{k+1,s}^{(1)}>0$.

Collecting \eqref{eq:fk-c-sens} and \eqref{eq:fk-s-sens} into one matrix identity, we have
\begin{equation}\label{eq:sens-matrix}
\mathbf{W}_k \mathbf{J}_k
= \!\left[\begin{smallmatrix}
\alpha_k f_{k+1,c}^{(0)} + \tfrac{k}{n}h_k^{(0)}
 & \alpha_k f_{k+1,s}^{(0)} \\[4pt]
\alpha_k f_{k+1,c}^{(1)}
 & \alpha_k f_{k+1,s}^{(1)} + \tfrac{k}{n}h_k^{(1)}
\end{smallmatrix}\right]\!.
\end{equation}
Taking determinants on both sides of \eqref{eq:sens-matrix} and using
$\det(\mathbf{W}_k \mathbf{J}_k) \!\!= \!\det \mathbf{W}_k\det \mathbf{J}_k$, the right-hand side expands to
\begin{align}
\det \mathbf{W}_k\,\det \mathbf{J}_k
={}& \alpha_k^2\,\det \mathbf{J}_{k+1}
+ \tfrac{\alpha_k k}{n}\,h_k^{(1)} f_{k+1,c}^{(0)} \notag\\
&+ \tfrac{\alpha_k k}{n}\,h_k^{(0)} f_{k+1,s}^{(1)}
+ \tfrac{k^2}{n^2}\,h_k^{(0)} h_k^{(1)}.
\label{eq:det-recursion}
\end{align}
At the boundary level $k = n$ we have $\alpha_n = 0$ and $\tfrac{k}{n} = 1$, so only the
last term of \eqref{eq:det-recursion} survives, giving
\begin{equation}\label{eq:det-base}
\det \mathbf{W}_n\,\det \mathbf{J}_n
= h_n^{(0)} h_n^{(1)} > 0,
\end{equation}
and since $\det \mathbf{W}_n > 0$, we obtain $\det \mathbf{J}_n > 0$. For the inductive step,
assume $\det \mathbf{J}_{k+1} > 0$. On the right-hand side of
\eqref{eq:det-recursion}, the first three terms are nonnegative, since
$\alpha_k \ge 0$ and every remaining factor is positive, and the last term is
strictly positive. Hence the right-hand side is strictly positive, and
$\det \mathbf{W}_k > 0$ yields $\det \mathbf{J}_k > 0$. By backward induction,
\eqref{eq:det-fk-pos} is positive for every $k$, and at $k = 1$ this yields
\begin{equation}\label{eq:det-cond}
f_{1,c}^{(0)} f_{1,s}^{(1)} > \bar{f}_{1,c}^{(1)}\,\bar{f}_{1,s}^{(0)}.
\end{equation}

\emph{Step 2 (Budget binding).} Suppose, for the sake of contradiction, that $(s^\star,c^\star)$ is optimal, satisfying $s^\star>0$ and $c^\star>0$, but we have $s^\star+c^\star<R$. If the  participation constraint were strict, that is $\UR(s^\star,c^\star,\lambda)>q\pi_0+\eta$, then by the continuity of $\UR$ there would exist $\varepsilon>0$ such that $(s^\star+\varepsilon,c^\star)$ remains feasible. Since $\frac{\partial \US}{\partial s} = \bar{f}_{1,s}^{(0)} + f_{1,s}^{(1)}>0$ by~\eqref{eq:abpr}, increasing $s^\star$ by $\varepsilon$ would strictly increase $\US$, contradicting optimality. Therefore, the  participation constraint must hold with equality at $(s^\star,c^\star)$, that is, \begin{equation}\label{eq:IC-binds} \UR(s^\star,c^\star,\lambda)=q\pi_0+\eta. \end{equation}
We now show that $ \frac{\partial \UR(s^\star,c^\star,\lambda)}{\partial s}\le 0$. Suppose for contradiction that $ \frac{\partial \UR(s^\star,c^\star,\lambda)}{\partial s}>0$. Since $\UR$ is differentiable, we would then have $\UR(s^\star\!+\!\varepsilon,c^\star,\lambda)>\UR(s^\star,c^\star,\lambda)=q\pi_0+\eta$ for all small $\varepsilon>0$, so $(s^\star\!+\!\varepsilon,c^\star)$ remains PC-feasible. Since also $\frac{\partial\US}{\partial s}>0$, this would again contradict optimality. We next show that $\frac{\partial \UR(s^\star,c^\star,\lambda)}{\partial c}>0$. Suppose, to the contrary, that $\frac{\partial \UR(s^\star,c^\star,\lambda)}{\partial c}\le 0$. Using \eqref{eq:grad-UR}, the two inequalities $\frac{\partial \UR}{\partial c}\le 0$ and $\frac{\partial\UR}{\partial s}\le 0$ become 
\begin{align} q  f_{1,c}^{(0)} &\le (1-q)\bar{f}_{1,c}^{(1)}, \label{eq:ineq1}\\ (1-q)f_{1,s}^{(1)}&\le q\bar{f}_{1,s}^{(0)}. \label{eq:ineq2} \end{align} Multiplying \eqref{eq:ineq1} and \eqref{eq:ineq2} yields $q(1-q)\, f_{1,c}^{(0)} f_{1,s}^{(1)} \le q(1-q)\,\bar{f}_{1,c}^{(1)}\,\bar{f}_{1,s}^{(0)}$. Since $q\in(0,1)$, this reduces to $ f_{1,c}^{(0)} f_{1,s}^{(1)} \le \bar{f}_{1,c}^{(1)}\,\bar{f}_{1,s}^{(0)}$, contradicting~\eqref{eq:det-cond}. Hence, we  have  $\frac{\partial \UR(s^\star,c^\star,\lambda)}{\partial c}>0$. As a result, we have established that, at $(s^\star,c^\star)$, $\frac{\partial \UR}{\partial c} \!=\! qf_{1,c}^{(0)}\!-\!(1\!-\!q)\bar{f}_{1,c}^{(1)} \!>\!0$ and $\frac{\partial \UR}{\partial s} \!=\! - q\bar{f}_{1,s}^{(0)}\!+\!(1\!-\!q)f_{1,s}^{(1)}\!\le\! 0$. Now consider the direction $\mathbf{d}=(\gamma,1)$ in the $(s,c)$ coordinate order, with $\gamma>0$. Using~\eqref{eq:grad-US}--\eqref{eq:grad-UR}, we get $\nabla \UR\cdot \mathbf{d} \!=\! \frac{\partial \UR}{\partial c}\!+\!\gamma \frac{\partial \UR}{\partial s}$ and $\nabla \US\cdot \mathbf{d} \!=\! -( f_{1,c}^{(0)}\!+\!\bar{f}_{1,c}^{(1)})\!+\!\gamma(\bar{f}_{1,s}^{(0)}\!+\!f_{1,s}^{(1)})$. To obtain a strict sender improvement, it suffices to require \begin{equation}\label{eq:gamma-min-def} \gamma>\gamma_{\min} \triangleq \frac{ f_{1,c}^{(0)}+\bar{f}_{1,c}^{(1)}}{\bar{f}_{1,s}^{(0)}+f_{1,s}^{(1)}}. \end{equation}
If $\frac{\partial \UR}{\partial s}<0$, then the strict PC improvement is guaranteed by \begin{equation}\label{eq:gamma-max-def} \gamma<\gamma_{\max} \triangleq - \frac{\frac{\partial \UR}{\partial c}}{\frac{\partial \UR}{\partial s}} =\frac{q  f_{1,c}^{(0)}-(1-q)\bar{f}_{1,c}^{(1)}}{q\bar{f}_{1,s}^{(0)}-(1-q)f_{1,s}^{(1)}}.
\end{equation} 
If $\frac{\partial \UR}{\partial s}=0$, then $\nabla \UR\cdot \mathbf{d}=\frac{\partial \UR}{\partial c}>0$ for every $\gamma>0$, so any $\gamma>\gamma_{\min}$ yields both sender and PC improvement. It remains to show that the interval $(\gamma_{\min},\gamma_{\max})$ is nonempty when $\frac{\partial \UR}{\partial s}<0$. A direct expansion of $\gamma_{\max} - \gamma_{\min}$ gives a strictly positive denominator and the numerator in the form of
\begin{align} 
&(\!\bar{f}_{1,s}^{(0)}\!\!+\!\!f_{1,s}^{(1)}\!)\bigl[\!q f_{1,c}^{(0)}\!\!-\!\!(1\!\!-\!\!q)\bar{f}_{1,c}^{(1)}\!\bigr] \!\!-\!\!(\! f_{1,c}^{(0)}\!+\!\bar{f}_{1,c}^{(1)}\!)\bigl[\!q\bar{f}_{1,s}^{(0)}\!-\!(1\!-\!q)f_{1,s}^{(1)}\!\bigr] \nonumber\\&=\!  f_{1,c}^{(0)}f_{1,s}^{(1)}\!-\!\bar{f}_{1,c}^{(1)}\,\bar{f}_{1,s}^{(0)}.\! \label{eq:cross-mult} 
\end{align}
Since $ f_{1,c}^{(0)} f_{1,s}^{(1)}-\bar{f}_{1,c}^{(1)}\,\bar{f}_{1,s}^{(0)}>0$ by~\eqref{eq:det-cond}, we obtain $\gamma_{\min}<\gamma_{\max}$. The two bounds impose opposing requirements on the same direction. The sender needs $s$ to grow fast enough relative to $c$, which is the lower bound, while the receiver needs $s$ not to grow too fast relative to $c$, which is the upper bound. Hence, when $\frac{\partial \UR}{\partial s}<0$, we may choose any $\gamma\in(\gamma_{\min},\gamma_{\max})$. For the chosen $\gamma$, both $\nabla \US\cdot \mathbf{d}>0$ and $\nabla \UR\cdot \mathbf{d}>0$ hold. Since $s^\star+c^\star<R$ and $\gamma>0$, a sufficiently small step from $(s^\star,c^\star)$ in the direction $\mathbf{d}$ remains inside the budget-feasible set. Because the PC is binding by~\eqref{eq:IC-binds} and its directional derivative along $\mathbf{d}$ is strictly positive, a sufficiently small step is also PC-feasible. This yields a feasible perturbation that strictly increases $\US$, contradicting the optimality of $(s^\star,c^\star)$. Hence, $s^\star+c^\star<R$ leads to a contradiction, and every interior optimal point must satisfy $s^\star+c^\star=R$.
\end{proof}

The sender exhausts its budget because state-$0$ and state-$1$ pushes move the two utilities in independent directions, which leaves a direction that raises the sender utility and the receiver utility at once whenever budget is left unspent. The hypothesis $s^\star,c^\star>0$ excludes the two endpoints where the budget is allocated entirely to one source state, and \cref{sec:boundary-anchoring} shows that neither of them is PC-feasible.

\subsection{Budget-Line Parameterization and Sender Optimality}\label{sec:BL-opt}
\vspace{-0.1cm}
Having established that the budget constraint is binding, the sender's problem reduces to a one-dimensional optimization problem. Thus, by setting $s = R - c$ with $c \in [0, R]$, we define
\begin{equation}\label{eq:BL-param}
\!\!\widetilde U_S(c,\lambda)\! := \!\US(R-c, c, \lambda), \;\;
\widetilde U_R(c,\lambda) \!:=\! \UR(R-c, c, \lambda),\!\!
\end{equation}
 and let $\Slack(c, \lambda) := \widetilde U_R(c, \lambda) - (q\pi_0 + \eta)$ denote the PC slack along the budget line. The chain rule together with Proposition~\ref{prop:signs-mono} gives $\frac{\partial \widetilde U_S}{\partial c} = \bigl(\frac{\partial  \US}{\partial c} - \frac{\partial  \US}{\partial s}\bigr)\big|_{s= R-c} = -(f_{1,c}^{(0)} + \bar{f}_{1,c}^{(1)} + \bar{f}_{1,s}^{(0)} + f_{1,s}^{(1)}) < 0$, so the sender utility is strictly decreasing along the budget line. This derivative is available on $(0,R)$, where both push rates are positive, and $\widetilde U_S(\cdot,\lambda)$ is continuous on $[0,R]$, so the strict decrease extends to the closed interval.

\begin{corollary}\label{cor:full-opt}
Fix $\lambda \ge 0$, $R > 0$, and $\eta > 0$. We define
\begin{align}\label{eq:IC-set-def}
 \mathcal{F}_\lambda := \{c \in [0, R] : \Slack(c, \lambda) \ge 0\}
\end{align}
as the budget-line feasible set. Assume $\mathcal{F}_\lambda$ is nonempty and that an optimal solution $(s^\star, c^\star)$ satisfies $s^\star > 0$ and $c^\star > 0$. Then, we have $s^\star + c^\star = R$, with $c^\star = \min \mathcal{F}_\lambda$ and $s^\star = R - c^\star$.
\end{corollary}

\begin{proof}
Budget binding follows from Theorem~\ref{thm:budget-binding}. Since $\widetilde U_S(\cdot,\lambda)$ is strictly decreasing, the optimum is attained at the smallest feasible $c$, and continuity of $\Slack(\cdot,\lambda)$ makes $\mathcal{F}_\lambda \subset [0, R]$ closed, so that minimum exists and is uniquely optimal.
\end{proof}

We denote $c_{\min}(\lambda) := \min \mathcal{F}_\lambda$. The sender therefore chooses the smallest state-$0$ push rate that satisfies the PC. The feasible set moves with the gossip rate, so the sender's optimum is settled only once the receivers' choice of $\lambda$ is known.

\section{Receiver's Best Response and Gossip Monotonicity}\label{sec:receiver}
\vspace{-0.1cm}

Given a sender policy $\theta = (s, c)$ on the budget line, the receiver team selects a gossip rate $\lambda \in [0, \Lambda]$ to maximize its utility. The difficulty is that gossip can first hurt the receivers by spreading packets that are fresh but inaccurate, before it helps by synchronizing nodes onto the freshest available update. We first recast the recursion in a form suited to the monotonicity analysis.

\subsection{The Transformed Recursion}\label{sec:defects}
\vspace{-0.1cm}

The proof recasts the recursion in~\eqref{eq:fk_recursion} in terms of the defect vector $\hvec_k$ of~\eqref{eq:defect-def}. The recursion below holds for any $s,c>0$, and only the sign arguments of \cref{sec:monotone} use $s>c$. Substituting $\fvec_k=[\pi_0\;\;\pi_1]^\top-\hvec_k$ into~\eqref{eq:fk_recursion} and using $q_{01}\pi_0=q_{10}\pi_1=\rho$, the stationary terms cancel and we obtain
\begin{equation}\label{eq:defect-recursion}
\mathbf{W}_k\hvec_k=\alpha_k\hvec_{k+1}+\rho\one,\;\; k=1,\ldots,n-1,
\end{equation}
with the boundary equation $\mathbf{W}_n\hvec_n=\rho\one$, where $\one=[1\;\;1]^\top$.

Both utilities depend on the recursion only through $\hvec_1$. Substituting $f_1^{(m)}=\pi_m-h_1^{(m)}$ into~\eqref{eq:util-f1} gives
\begin{align}
\UR &= q\pi_0+(1-q)\pi_1-\big(qh_1^{(0)}+(1-q)h_1^{(1)}\big),\label{eq:UR-defect}\\
\US &= \pi_1+\big(h_1^{(0)}-h_1^{(1)}\big).\label{eq:US-defect}
\end{align}
The receiver utility is governed by $qh_1^{(0)}+(1-q)h_1^{(1)}$ and the sender utility is governed by $h_1^{(0)}-h_1^{(1)}$, since the remaining terms are constants. These are the only two combinations of $\hvec_1$ that the utilities see, so we take them as the rows of a change of coordinates,
\begin{equation}\label{eq:T-def}
\mathbf{T}=\bmat{q & 1-q\\ 1 & -1},\;\; \mathbf{T}^{-1}=\bmat{1 & 1-q\\ 1 & -q},
\end{equation}
and define the transformed defect $\mathbf{g}_k=\mathbf{T}\hvec_k$, with components $g_k^{(1)}=qh_k^{(0)}+(1-q)h_k^{(1)}$ and $g_k^{(2)}=h_k^{(0)}-h_k^{(1)}$. The first carries the receiver-weighted defect and the second carries the mode gap that the sender cares about, so each utility depends on one component. Differentiating~\eqref{eq:UR-defect} and~\eqref{eq:US-defect} at $k=1$ gives
\begin{equation}\label{eq:util-g}
\frac{\partial \UR}{\partial\lambda}=-\frac{\partial g_1^{(1)}}{\partial\lambda},\;\;
\frac{\partial \US}{\partial\lambda}=\frac{\partial g_1^{(2)}}{\partial\lambda},
\end{equation}
so both monotonicity claims reduce to showing that $\mathbf{g}_1$ is nonincreasing in $\lambda$ component-wise.

Applying $\mathbf{T}$ to~\eqref{eq:defect-recursion} and using $\mathbf{T}\one=[1\;\;0]^\top$ gives the transformed recursion
\begin{equation}\label{eq:transformed-recursion}
\mathbf{M}_k=\mathbf{T}\mathbf{W}_k\mathbf{T}^{-1},\;\; \mathbf{M}_k \mathbf{g}_k=\alpha_k \mathbf{g}_{k+1}+\rho\eone
\end{equation}
for $k=1,\ldots,n-1$, with the boundary equation $\mathbf{M}_n \mathbf{g}_n=\rho\eone$. Using the definition of $\mathbf{W}_k$ in~\eqref{eq:Wk_vk_def}, the conjugated matrix has the form
\begin{align}\label{eq:Mk-explicit}
\mathbf{M}_k &= \bmat{q_{01}+q_{10} & -\Delta\\ 0 & 0}+\alpha_k \mathbf{I}\notag\\
&\;+\frac{k}{n}\bmat{qc+(1-q)s & -q(1-q)(s-c)\\ -(s-c) & (1-q)c+qs}.
\end{align}
When $s>c$, and $\Delta>0$ under Assumption~\ref{as:outside}, the matrix $\mathbf{M}_k$ in~\eqref{eq:Mk-explicit} has strictly positive diagonal entries and strictly negative off-diagonal entries. Since conjugation preserves determinants, $\det \mathbf{M}_k=\det \mathbf{W}_k>0$. Reading the explicit inverse in the proof of Lemma~\ref{lem:sign} off this sign pattern, the four entries of $\mathbf{M}_k^{-1}$ are the two diagonal entries of $\mathbf{M}_k$ and the negatives of its two off-diagonal entries, each divided by $\det \mathbf{M}_k$, so
\begin{equation}\label{eq:Mk-inv-nonneg}
\mathbf{M}_k^{-1}> 0\;\;\text{component-wise},\;\; k=1,\ldots,n.
\end{equation}
The regime condition $s>c$ is used in \cref{sec:monotone} only for the two off-diagonal entries of the last term of~\eqref{eq:Mk-explicit}, which make $\mathbf{M}_k$ a nonsingular M-matrix and give~\eqref{eq:Mk-inv-nonneg}. Once~\eqref{eq:Mk-inv-nonneg} holds, the two inductions that follow use nothing else about the parameters.

\subsection{Gossip Monotonicity in the Strategic Regime}\label{sec:monotone}
\vspace{-0.1cm}

We now state the two monotonicity results and prove them together.

\begin{theorem}\label{thm:UR-monotone}
Let \cref{as:outside} hold and assume $s>c>0$. Then, for all $\lambda\ge 0$ and all $n\ge 2$,
\begin{equation}\label{eq:gossip-mono}
\frac{\partial \UR}{\partial\lambda}> 0,\;\; \frac{\partial \US}{\partial\lambda}< 0.
\end{equation}
\end{theorem}

\begin{proof}
By~\eqref{eq:util-g}, it suffices to show $\frac{\partial \mathbf{g}_1}{\partial \lambda}< 0$ component-wise. We define the level increment $\bm{\delta}_k=\mathbf{g}_k-\mathbf{g}_{k+1}$ for $k=1,\ldots,n-1$, which measures the change in the transformed defect between adjacent subset sizes. Since $\hvec_k\!-\!\hvec_{k+1}\!=\!\fvec_{k+1}\!-\!\fvec_k$ by~\eqref{eq:defect-def}, $\bm{\delta}_k$ is the accuracy gain from enlarging the subset, expressed in the coordinates of~\eqref{eq:T-def}. The proof runs two backward inductions, the first signing $\bm{\delta}_k$ and the second using that sign to control $\frac{\partial \mathbf{g}_k}{\partial\lambda}$. Both use $\mathbf{M}_k^{-1}> 0$ from~\eqref{eq:Mk-inv-nonneg}.

\emph{Step 1 (Level increments).} We show that $\bm{\delta}_k> 0$ component-wise for all $k=1,\ldots,n-1$ and all $\lambda\ge 0$. From the definition of $\mathbf{W}_k$ in~\eqref{eq:Wk_vk_def}, we have $\mathbf{W}_{k+1}-\mathbf{W}_k=(\alpha_{k+1}-\alpha_k)\mathbf{I}+\diag\!\big(\frac{c}{n},\frac{s}{n}\big)$. Applying the defect recursion~\eqref{eq:defect-recursion} at levels $k$ and $k+1$ and eliminating $\mathbf{W}_k\hvec_{k+1}$ with this identity, we obtain
\begin{align}\label{eq:gap-recursion}
\mathbf{W}_k(\hvec_k-\hvec_{k+1}) &= \alpha_{k+1}(\hvec_{k+1}-\hvec_{k+2})+\bmat{\frac{c}{n}h_{k+1}^{(0)}\\[3pt] \frac{s}{n}h_{k+1}^{(1)}},
\end{align}
for $k=1,\ldots,n-1$, where the term with $\alpha_{k+1}$ is absent at $k=n-1$ since $\alpha_n=0$. Multiplying~\eqref{eq:gap-recursion} by $\mathbf{T}$ and using $\bm{\delta}_k=\mathbf{T}(\hvec_k-\hvec_{k+1})$ gives
\begin{align}\label{eq:delta-recursion}
\mathbf{M}_k\bm{\delta}_k &= \alpha_{k+1}\bm{\delta}_{k+1}+\frac{1}{n}\bmat{qc\,h_{k+1}^{(0)}+(1-q)s\,h_{k+1}^{(1)}\\[3pt] c\,h_{k+1}^{(0)}-s\,h_{k+1}^{(1)}}.
\end{align}
The first component of the forcing vector in~\eqref{eq:delta-recursion} is strictly positive because $\hvec_{k+1}>0$ by~\eqref{eq:hk-strict}. The second component subtracts one positive quantity from another, so its sign is not immediate. To sign it, we subtract the two rows of the defect recursion, which turns that term into a multiple of $\delta_j^{(2)}$, the quantity the induction already controls. At level $j$, this gives
\begin{equation}\label{eq:row-diff}
\frac{1}{n}\big(c\,h_j^{(0)}-s\,h_j^{(1)}\big)=-\frac{\alpha_j}{j}\,\delta_j^{(2)},\;\; j=1,\ldots,n-1,
\end{equation}
while the same subtraction on the boundary equation $\mathbf{W}_n\hvec_n=\rho\one$ gives $c\,h_n^{(0)}-s\,h_n^{(1)}=0$.

Next, we prove $\bm{\delta}_k> 0$ by backward induction. At $k=n-1$, the term $\alpha_n=0$ removes $\bm{\delta}_n$ from~\eqref{eq:delta-recursion}, and the boundary identity makes the second component of the forcing vanish, so the forcing has the sign pattern $(+,0)$. With $\mathbf{M}_{n-1}^{-1}>0$ from~\eqref{eq:Mk-inv-nonneg}, we obtain $\bm{\delta}_{n-1}>0$. Assume now $\bm{\delta}_{k+1}> 0$ for some $k\le n-2$. The first component of the right-hand side of~\eqref{eq:delta-recursion} equals $\alpha_{k+1}\delta_{k+1}^{(1)}+\frac{1}{n}\big(qc\,h_{k+1}^{(0)}+(1-q)s\,h_{k+1}^{(1)}\big)>0$. The second component equals $\alpha_{k+1}\delta_{k+1}^{(2)}+\frac{1}{n}\big(c\,h_{k+1}^{(0)}-s\,h_{k+1}^{(1)}\big)$, and substituting~\eqref{eq:row-diff} at $j=k+1$ reduces it to $\frac{k}{k+1}\alpha_{k+1}\delta_{k+1}^{(2)}\ge 0$. The right-hand side of~\eqref{eq:delta-recursion} is therefore positive in its first entry and nonnegative in its second, and $\mathbf{M}_k^{-1}>0$ gives $\bm{\delta}_k>0$. In the original coordinates, this says that the receiver-weighted accuracy increment between adjacent subset sizes and the mode gap of that increment are both positive when $s>c$.

\emph{Step 2 (Gossip sensitivity).} Since $\alpha_k=\frac{k(n-k)}{n-1}\lambda$ and $\mathbf{W}_k$ depends on $\lambda$ only through the term $\alpha_k \mathbf{I}$, the matrix $\mathbf{M}_k$ depends on $\lambda$ only through $\alpha_k \mathbf{I}$ as well, so $\frac{\partial \mathbf{M}_k}{\partial\lambda}=\frac{k(n-k)}{n-1}\mathbf{I}$. Differentiating~\eqref{eq:transformed-recursion} with respect to $\lambda$ and using $\mathbf{g}_k-\mathbf{g}_{k+1}=\bm{\delta}_k$ gives
\begin{equation}\label{eq:g-prime-rec}
\mathbf{M}_k\frac{\partial \mathbf{g}_k}{\partial \lambda}=-\frac{k(n-k)}{n-1}\,\bm{\delta}_k+\alpha_k\frac{\partial \mathbf{g}_{k+1}}{\partial \lambda}.
\end{equation}
At the boundary $k=n$, both $\mathbf{M}_n$ and $\rho\eone$ are independent of $\lambda$ since $\alpha_n=0$, so $\frac{\partial \mathbf{g}_n}{\partial \lambda}=0$. Assume $\frac{\partial \mathbf{g}_{k+1}}{\partial \lambda}\le 0$. By Step~1 we have $\bm{\delta}_k>0$, and $\frac{k(n-k)}{n-1}>0$ for $1\le k\le n-1$, so the first term of~\eqref{eq:g-prime-rec} is strictly negative. The second term is nonpositive because $\alpha_k\ge 0$. Since $\mathbf{M}_k^{-1}>0$ by~\eqref{eq:Mk-inv-nonneg}, we obtain $\frac{\partial \mathbf{g}_k}{\partial \lambda}<0$. Backward induction to $k=1$ gives $\frac{\partial \mathbf{g}_1}{\partial \lambda}<0$, and~\eqref{eq:util-g} yields~\eqref{eq:gossip-mono}.
\end{proof}

The two inequalities in~\eqref{eq:gossip-mono} point in opposite directions, so at a fixed policy with $s>c$ the sender never gains from additional gossip while the receivers always do. The receiver utility is therefore strictly increasing in $\lambda$, and $\lambda^\star=\Lambda$ is the unique best response at every network size.

Outside this regime the shape of $\UR(\theta,\cdot)$ is not settled by \cref{thm:UR-monotone}, but its endpoints are. Setting $\lambda=0$ in~\eqref{eq:fk_recursion} decouples the levels, while letting $\lambda$ grow drives every $\fvec_k$ to the boundary solution $\fvec_n$ of~\eqref{eq:fn_closed}, which does not depend on $\lambda$. Let $D_k = q_{01}s + c\,q_{10} + \frac{k\,c\,s}{n}$ extend the boundary quantity $D_n$ of~\eqref{eq:fn_closed} to every level. For $s,c>0$, solving the two limits and subtracting gives
\begin{equation}\label{eq:endpoint-ordering}
\!\UR(\theta,\infty)\! -\! \UR(\theta,0)
\!= \!\frac{c\,q_{01}q_{10}s(n\!-\!1)\bigl(c(1\!-\!q)\!+\!qs\bigr)}{n(q_{01}\!+\!q_{10})\,D_n\,D_1} \!\!>\!\! 0,\!\!
\end{equation}
Gossip is therefore beneficial to the receivers at the two extremes whatever the policy. The best response can fall strictly inside $[0,\Lambda]$ only if $\UR$ rises to an interior peak and then falls away from it, which requires the derivative to turn from positive to negative.

To examine whether such a turn can occur, we evaluated the derivative over the sweep, which covered $134{,}217{,}728$ parameter combinations satisfying Assumption~1. We used $n\in\{2,3,5,8,\allowbreak 12,18,28,44,\allowbreak 68,106,165,257,\allowbreak 399,621,965,1500\}$, $16$ equally spaced values of $q\in[0.02,0.98]$, and $32$ logarithmically spaced values of each of $q_{01}$, $q_{10}$, $s$, and $c$ over $[10^{-3},10^4]$. For each combination, we sampled the derivative at $512$ logarithmically spaced values of $\lambda\in[10^{-5},10^5]$. The derivative changed sign at most once, and every change was from negative to positive. Note that for $n=2$ the top level $\fvec_2$ is the boundary vector of~\eqref{eq:fn_closed}, and the conjecture below can be verified in closed form.

\begin{conjecture}\label{conj:one_dip}
Let \cref{as:outside} hold. For all $n > 1$, $0 < q < 1$, and strictly positive $q_{01}$, $q_{10}$, $s$, and $c$, there is a $\lambda_{\mathrm d}\in[0,\infty)$ such that
\begin{equation}\label{eq:one-dip}
\frac{\partial \UR(\theta,\lambda)}{\partial\lambda}<0 \;\;\text{for}\;\; \lambda<\lambda_{\mathrm d},\;\;
\frac{\partial \UR(\theta,\lambda)}{\partial\lambda}>0 \;\;\text{for}\;\; \lambda>\lambda_{\mathrm d}.
\end{equation}
\end{conjecture}

Thus $\UR(\theta,\cdot)$ falls up to $\lambda_{\mathrm d}$ and rises after it, with $\lambda_{\mathrm d}=0$ covering the case where it rises throughout. A function of this shape attains its maximum over $[0,\Lambda]$ only at an endpoint, so under \cref{conj:one_dip} the receiver's best response is $\lambda^\star\in\{0,\Lambda\}$ at every network size and every followed policy with $s,c>0$. The conjecture is used only where $c>s$, which \cref{thm:UR-monotone} does not cover and which is not the sender's preferred operating regime. The remaining case $s=c$ is handled in \cref{sec:regimes}.

\subsection{Sender Utility in the Remaining Regimes}\label{sec:regimes}
\vspace{-0.1cm}

The complementary regime $c>s$ and the symmetric case $c=s$ follow from a symmetry of the model. Consider the state relabeling $Q^{\mathrm{rel}}(t)=1-Q(t)$. The relabeled system is an instance of the same model in which the source rates $q_{01}$ and $q_{10}$ are interchanged and the push rates $s$ and $c$ are interchanged, while the gossip mechanism is unchanged since gossip decisions depend only on version age. An original policy with $c>s$ therefore maps to a relabeled policy whose state-$1$ push rate exceeds its state-$0$ push rate. Because a receiver holding relabeled state~$1$ holds the original state~$0$, the sender utility of the relabeled system satisfies $U_S^{\mathrm{rel}} = 1-\US(s,c,\lambda)$, and differentiating with respect to $\lambda$ gives
\begin{equation}\label{eq:US-relabel-derivative}
    \frac{\partial U_S^{\mathrm{rel}}}{\partial \lambda} = -\frac{\partial \US}{\partial \lambda}.
\end{equation}
We also record that $\US$ depends on the push rates, the source rates, and $\lambda$ through $\mathbf f_1$ alone, so it does not depend on the receiver weight $q$. This relabeling lets us cover all three regimes of the push rates in the following corollary.

\begin{corollary}\label{cor:direct-sign}
For every fixed $s,c\!>\!0$ and every $\lambda\!\ge\!0$, we have
\begin{align}
    (c-s)\frac{\partial \US(s,c,\lambda)}{\partial \lambda}&\ge 0,\label{eq:direct-sign}\\
    \sgn\bigl(\US(s,c,\lambda)-\pi_1\bigr)&=\sgn(s-c).\label{eq:US-pi1-order}
\end{align}
\end{corollary}

\begin{proof}
We treat the three cases separately. When $s>c$, the sender-side monotonicity in Theorem~\ref{thm:UR-monotone} gives $\tfrac{\partial \US}{\partial\lambda}< 0$, and since $c-s<0$, inequality in \eqref{eq:direct-sign} follows. For \eqref{eq:US-pi1-order}, the boundary identity $c\,h_n^{(0)}=s\,h_n^{(1)}$ obtained in the proof of Theorem~\ref{thm:UR-monotone} gives $g_n^{(2)}=\tfrac{s-c}{s}h_n^{(0)}>0$, and Step~1 of that proof gives $\bm{\delta}_k>0$ for every $k$, so $g_1^{(2)}\ge g_n^{(2)}>0$. Since $\US=\pi_1+g_1^{(2)}$ by~\eqref{eq:US-defect}, we obtain $\US>\pi_1$.

When $c>s$, the relabeled system has a state-$1$ push rate larger than its state-$0$ push rate, so Theorem~\ref{thm:UR-monotone} applies to it once \cref{as:outside} holds for it. Written in the original rates, that the assumption becomes $q^{\mathrm{rel}}q_{01}>(1-q^{\mathrm{rel}})\,q_{10}$, where $q^{\mathrm{rel}}$ is the receiver weight of the relabeled system, so it holds for every $q^{\mathrm{rel}}\in(\pi_0,1)$. Since $U_S^{\mathrm{rel}}$ does not depend on $q^{\mathrm{rel}}$, we may evaluate it at such a weight, and Theorem~\ref{thm:UR-monotone} gives $\tfrac{\partial U_S^{\mathrm{rel}}}{\partial\lambda}< 0$. Substituting~\eqref{eq:US-relabel-derivative} produces $\tfrac{\partial \US}{\partial\lambda}> 0$, and since $c-s>0$, inequality in \eqref{eq:direct-sign} again holds. The first case applied to the relabeled system gives $U_S^{\mathrm{rel}}>\pi_0$, since the relabeled source is in state~$1$ with probability $\pi_0$, and $\US=1-U_S^{\mathrm{rel}}$ then gives $\US<\pi_1$.

When $c=s$, subtracting the two rows of the defect recursion in \eqref{eq:defect-recursion} gives
\begin{equation}\label{eq:g2-symmetric-rec}
    \left(\alpha_k+\frac{kc}{n}\right)g_k^{(2)} = \alpha_k\, g_{k+1}^{(2)},
\end{equation}
for $k=1,\ldots,n-1$, and at the boundary $k=n$, where $\alpha_n=0$, the same subtraction yields $c\,g_n^{(2)}=0$, so $g_n^{(2)}=0$. The backward recursion in \eqref{eq:g2-symmetric-rec} then forces $g_k^{(2)}=0$ for every $k$. Since $\US=\pi_1+g_1^{(2)}$ by~\eqref{eq:US-defect}, we obtain $\US=\pi_1$, which is independent of $\lambda$. Combining the three cases completes the proof.
\end{proof}
Gossip is content-agnostic, propagating the freshest available packet regardless of its content, so it does not favor either source state. When $s>c$, however, state-$1$ packets enter the network more frequently. After the source switches from state~$1$ to state~$0$, many nodes may continue to hold stale state-$1$ packets and thus declare the sender-preferred state. Once a current state-$0$ packet reaches a node, a higher gossip rate disseminates that corrective packet more rapidly, reducing the persistence of stale state-$1$ declarations. The proof also shows that~\eqref{eq:direct-sign} is strict whenever $s\ne c$. The ordering in~\eqref{eq:US-pi1-order} captures the same policy asymmetry at the utility level, and the equilibrium analysis in~\cref{sec:equilibrium} uses it to compare policies across the different rate regimes.

The symmetric case also completes the receiver side of \cref{thm:UR-monotone}, which assumes $s>c>0$ strictly. At $s=c$ the lower-left entry of the last term of~\eqref{eq:Mk-explicit} vanishes, so $\mathbf{M}_k$ and $\mathbf{M}_k^{-1}$ are upper triangular with positive entries on and above the diagonal, and~\eqref{eq:Mk-inv-nonneg} weakens to $\mathbf{M}_k^{-1}\!\ge\!0$. The proof above gives $g_k^{(2)}\!=\!0$ at every level, so the forcing vectors of~\eqref{eq:delta-recursion} and~\eqref{eq:g-prime-rec} take the forms $(+,0)$ and $(-,0)$, which an upper triangular inverse preserves. Both inductions therefore close, and $\frac{\partial \UR}{\partial\lambda}>0$ holds for every $s\ge c>0$.

\section{Stackelberg Equilibrium and the Effect of Gossip}\label{sec:equilibrium}
\vspace{-0.1cm}

We now combine the budget binding of Theorem~\ref{thm:budget-binding}, the budget-line monotonicity of Corollary~\ref{cor:full-opt}, and the gossip monotonicity of Theorem~\ref{thm:UR-monotone} to characterize the Stackelberg equilibrium.

\subsection{Boundary Characterization of the PC Slack}\label{sec:boundary-anchoring}
\vspace{-0.1cm}
We first identify the two endpoints of the budget line, working throughout with the budget-line PC slack $\Slack(c,\lambda)$ defined below~\eqref{eq:BL-param} and the PC-feasible set $\mathcal{F}_\lambda$ of~\eqref{eq:IC-set-def}. Neither endpoint is PC-feasible. At $c=0$ the sender transmits only in state~$1$, so a node that follows is accurate only while $Q=1$ and obtains $(1-q)\pi_1$, which is below $q\pi_0$ by \cref{as:outside}. At $c=R$ the sender allocates its entire budget to state-$0$ information, so a node that follows obtains exactly the default utility $q\pi_0$, which falls short of $q\pi_0+\eta$. In both cases the receiver team defaults and the sender obtains zero utility.

By Corollary~\ref{cor:full-opt}, the sender-optimal feasible policy is $c^\star=c_{\min}(\lambda)$ whenever $\mathcal F_\lambda$ is nonempty, and the two endpoints place it in the interior, $0<c^\star<R$. Since $c^\star$ is the smallest feasible point and the slack is continuous, the PC binds there, so $\widetilde U_R(c^\star,\lambda)=q\pi_0+\eta$.

\subsection{Equilibrium Characterization}\label{sec:SE}
\vspace{-0.1cm}

\begin{theorem}\label{thm:SE}
Let \cref{as:outside} hold, and let $R>0$, $\Lambda>0$, and $\eta>0$.
\begin{enumerate}[label=(\alph*)]
\item An optimistic Stackelberg equilibrium exists.
\item Let the feasible set $\mathcal{F}_\Lambda$ of~\eqref{eq:IC-set-def} be nonempty with $c_{\min}(\Lambda)<\tfrac{R}{2}$. Then the equilibrium is unique and is given by
\begin{equation}\label{eq:SE_policy}
(s^\star,c^\star)=\bigl(R-c_{\min}(\Lambda),\;c_{\min}(\Lambda)\bigr),\;\;\lambda^\star=\Lambda.
\end{equation}
\item Under \cref{conj:one_dip}, every followed policy with $s,c>0$ induces $\lambda^\star\in\{0,\Lambda\}$.
\end{enumerate}
\end{theorem}

\begin{proof}
\emph{Part (a).} By \cref{thm:backward_recursion}, every $\mathbf{W}_k$ is invertible on $\Theta$ and its entries are affine in $\theta$, so $\UR$ and $\US$ are continuous on $\Theta\times[0,\Lambda]$. Let $\Theta_{\mathrm{PC}}=\{\theta\in\Theta:\UR(\theta,\lambda)\ge q\pi_0+\eta \text{ for some }\lambda\in[0,\Lambda]\}$ collect the policies that the receivers are willing to follow. Every other policy is met by the default response and pays the sender zero. We first show that $\Theta_{\mathrm{PC}}$ is compact.

Boundedness is inherited from $\Theta$. For closedness, let $\theta_\ell\to\theta\in\Theta$ with each $\theta_\ell$ satisfying the PC at some $\lambda_\ell$. The interval $[0,\Lambda]$ is compact, so a subsequence of $\lambda_\ell$ converges to some $\check\lambda$, and continuity gives $\UR(\theta,\check\lambda)\ge q\pi_0+\eta$. The origin is the only limit point of $\Theta$ that $\Theta$ omits, so it remains to keep $\Theta_{\mathrm{PC}}$ away from it. Adding $\pi_1$ times the first row of~\eqref{eq:fk_recursion} to $\pi_0$ times its second row and using $q_{01}\pi_0=q_{10}\pi_1=\rho$ gives
\begin{align}\label{eq:cross-recursion}
(q_{01}&+q_{10}+\alpha_k)\bigl(\pi_1 f_k^{(0)}+\pi_0 f_k^{(1)}\bigr)=\rho+\pi_1\alpha_k^{(0)}h_k^{(0)}\notag\\
&+\pi_0\alpha_k^{(1)}h_k^{(1)}+\alpha_k\bigl(\pi_1 f_{k+1}^{(0)}+\pi_0 f_{k+1}^{(1)}\bigr),
\end{align}
for $k=1,\ldots,n-1$, and without the last term at $k=n$. The two push terms are at most $\pi_0\pi_1(s+c)$, since~\eqref{eq:alpha_def} gives $\alpha_k^{(0)}\le c$ and $\alpha_k^{(1)}\le s$, while~\eqref{eq:defect-def} gives $h_k^{(0)}\le\pi_0$ and $h_k^{(1)}\le\pi_1$. Using this and $\rho=(q_{01}+q_{10})\pi_0\pi_1$ in~\eqref{eq:cross-recursion} at $k=n$ gives $\pi_1 f_n^{(0)}+\pi_0 f_n^{(1)}\le\pi_0\pi_1\bigl(1+\tfrac{s+c}{q_{01}+q_{10}}\bigr)$, and inserting that bound on the right-hand side at level $n-1$ returns the same bound one level down, so backward induction carries it to $k=1$. Eliminating $f_1^{(0)}$ between the bound at $k=1$ and~\eqref{eq:util-f1} leaves $f_1^{(1)}$ with the coefficient $(1-q)-\tfrac{q\pi_0}{\pi_1}$, which is negative by \cref{as:outside}, so dropping that term gives
\begin{equation}\label{eq:UR-budget-bound}
\UR(\theta,\lambda)\le q\pi_0\Bigl(1+\frac{s+c}{q_{01}+q_{10}}\Bigr),
\end{equation}
for every $\theta\in\Theta$ and $\lambda\ge 0$. The PC therefore forces $s+c\ge\tfrac{\eta(q_{01}+q_{10})}{q\pi_0}$, which keeps $\Theta_{\mathrm{PC}}$ away from the origin and makes it compact.

If $\Theta_{\mathrm{PC}}$ is empty, every policy pays the sender zero and every policy is an equilibrium. Otherwise, the receivers maximize the continuous function $\UR(\theta,\cdot)$ over the fixed compact interval $[0,\Lambda]$, so Berge's maximum theorem makes their best-response correspondence nonempty, compact-valued, and upper hemicontinuous on $\Theta_{\mathrm{PC}}$. The optimistic sender payoff maximizes $\US(\theta,\cdot)$ over that correspondence, so it is upper semicontinuous on $\Theta_{\mathrm{PC}}$ and attains its maximum there. That maximum is nonnegative, while every policy outside $\Theta_{\mathrm{PC}}$ pays zero, so it is a maximum over $\Theta$.

\emph{Part (b).} By \cref{thm:UR-monotone}, $\UR(\theta,\cdot)$ is strictly increasing whenever $s>c>0$. Such a policy is followed exactly when it is PC-feasible at $\lambda=\Lambda$, and the receivers then answer it with $\lambda^\star=\Lambda$. The policy $\bar\theta=\bigl(R-c_{\min}(\Lambda),c_{\min}(\Lambda)\bigr)$ is of this kind, since $c_{\min}(\Lambda)>0$ by \cref{sec:boundary-anchoring} and $R-c_{\min}(\Lambda)>c_{\min}(\Lambda)$ follows from $c_{\min}(\Lambda)<\tfrac{R}{2}$, and it pays $\US(\bar\theta,\Lambda)>\pi_1$ by~\eqref{eq:US-pi1-order}.

Every policy outside the class of followed policies with $s>c>0$ yields a sender utility of at most $\pi_1$. A policy that fails the PC at every $\lambda\in[0,\Lambda]$ is met by $\sigma_{\mathrm{default}}$ of~\eqref{eq:BR} and pays the sender $0$. Every policy with $c=0$ is of this kind, since then no node holds a state-$0$ packet and $\UR\le(1-q)\pi_1<q\pi_0$ by \cref{as:outside}. A followed policy with $s=0$ also pays $0$, because the sender never transmits in state~$1$ and no node ever holds a state-$1$ packet. The only remaining policies are followed policies satisfying $0<s\le c$. By~\eqref{eq:US-pi1-order}, each such policy yields a sender utility of at most $\pi_1$ for every gossip rate.

An equilibrium $\theta^\star$ exists by part~(a) and pays at least $\US(\bar\theta,\Lambda)>\pi_1$, so it is followed, satisfies $s^\star>c^\star>0$, and induces $\lambda^\star=\Lambda$. It also maximizes $\US(\cdot,\Lambda)$ among the policies that are PC-feasible at $\Lambda$, because those with $s\le c$ pay at most $\pi_1$ while those with $s>c$ induce $\Lambda$ themselves. \cref{cor:full-opt} applies at $\theta^\star$ and gives $s^\star+c^\star=R$ with $c^\star=c_{\min}(\Lambda)$, which fixes $\theta^\star$ and makes the equilibrium unique.

\emph{Part (c).} Under~\eqref{eq:one-dip}, $\UR(\theta,\cdot)$ decreases up to $\lambda_{\mathrm d}$ and increases after it, so it attains its maximum over $[0,\Lambda]$ at an endpoint.
\end{proof}

Under the conditions of part~(b), Theorem~\ref{thm:SE} reduces the equilibrium computation to locating the smallest PC-feasible point on the one-dimensional budget line. This characterization does not require $\mathcal{F}_\Lambda$ to be connected and holds for every $n\ge 2$ without invoking \cref{conj:one_dip}. The two conditions of part~(b) hold together exactly when $\mathcal{F}_\Lambda$ meets $[0,\tfrac{R}{2})$, so the hypothesis asks only that some policy on the strategic half of the budget line be PC-feasible at the gossip cap, and the equilibrium policy then satisfies $s^\star>c^\star$.

When $\mathcal{F}_\Lambda$ is nonempty but $c_{\min}(\Lambda)\ge\tfrac{R}{2}$, part~(b) does not apply and the equilibrium leaves the strategic half. \cref{cor:full-opt} still places the sender at $c^\star=c_{\min}(\Lambda)$ on the budget line, and~\eqref{eq:US-pi1-order} then bounds the payoff in this regime by $\pi_1$. Furthermore, part~(c) pins the realized gossip rate to endpoints of the gossip rate set, i.e. $\lambda^\star \in \{0,\Lambda\}$, under \cref{conj:one_dip}. At the default parameters of \cref{sec:numerics} with $R=20$, this case occupies the short interval $\Lambda\in(4.92,5.43)$, and $c_{\min}(\Lambda)$ falls to $\tfrac{R}{2}$ at its right endpoint.

\subsection{Who Benefits from Gossip?}\label{sec:benefits}
\vspace{-0.1cm}

At a fixed policy with $s>c$, gossip hurts the sender, but by raising the receiver utility it also relaxes the PC. On the budget line $s=R-c$, we define
\begin{equation}\label{eq:Fstr}
    \mathcal F_{\lambda}^{\mathrm{str}}
    =
    \bigl\{c\in[0,\tfrac{R}{2}): \Slack(c,\lambda)\ge 0\bigr\},
\end{equation}
and let $c_{\min}^{\mathrm{str}}(\lambda)=\min \mathcal F_{\lambda}^{\mathrm{str}}$ whenever the set is nonempty.

\begin{corollary}\label{cor:cmin-lambda-strategic}
Let $n\ge 2$, $R>0$, and $\eta> 0$ be fixed. On the strategic half of the budget line, $c\in[0,\tfrac{R}{2})$, the PC-feasible sets are nested in the realized gossip rate. That is, if $0\le \lambda_1\le \lambda_2$, then
\begin{equation}\label{eq:set-nesting-str}
    \mathcal F_{\lambda_1}^{\mathrm{str}}
    \subseteq
    \mathcal F_{\lambda_2}^{\mathrm{str}}.
\end{equation}
Consequently, whenever both minima are well defined, we have
\begin{equation}\label{eq:cmin-str-mono}
    c_{\min}^{\mathrm{str}}(\lambda_2)
    \le
    c_{\min}^{\mathrm{str}}(\lambda_1).
\end{equation}
\end{corollary}

\begin{proof}
Let $c\in\mathcal F_{\lambda_1}^{\mathrm{str}}$, so that $\widetilde U_R(c,\lambda_1)\ge q\pi_0+\eta$. Feasibility forces $c>0$, since $c=0$ gives $\widetilde U_R=(1-q)\pi_1<q\pi_0$ by \cref{sec:boundary-anchoring}. Since $c<\tfrac{R}{2}$, we have $s=R-c>c$, and Theorem~\ref{thm:UR-monotone} gives $\widetilde U_R(c,\lambda_2)\ge\widetilde U_R(c,\lambda_1)\ge q\pi_0+\eta$. Thus, we have $c\in\mathcal F_{\lambda_2}^{\mathrm{str}}$, proving~\eqref{eq:set-nesting-str}, and~\eqref{eq:cmin-str-mono} follows from this set inclusion.
\end{proof}

The two effects act on the equilibrium payoff at once. Writing the equilibrium sender utility on the strategic half as $\widetilde U_S(c_{\min}^{\mathrm{str}}(\lambda),\lambda)$ and differentiating, wherever $c_{\min}^{\mathrm{str}}$ is differentiable, gives
\begin{equation}\label{eq:US-eq-decomp}
\frac{d}{d\lambda}\,\widetilde U_S\!\left(c_{\min}^{\mathrm{str}}(\lambda),\lambda\right)
=\underbrace{\frac{\partial\widetilde U_S}{\partial\lambda}}_{\text{direct}}
+\underbrace{\frac{\partial\widetilde U_S}{\partial c}\frac{dc_{\min}^{\mathrm{str}}}{d\lambda}}_{\text{policy adjustment}}.
\end{equation}
The first term is negative by \cref{thm:UR-monotone}. The second is nonnegative, since $\frac{\partial\widetilde U_S}{\partial c}<0$ by the chain rule and \cref{prop:signs-mono}, and $c_{\min}^{\mathrm{str}}$ is nonincreasing by \cref{cor:cmin-lambda-strategic}. The receiver's fixed-policy gain need not survive into the equilibrium payoff, since the sender picks the smallest PC-feasible policy and the receiver utility is pinned at $q\pi_0+\eta$. We determine which of the two effects dominates numerically in \cref{sec:numerics}.

\section{Numerical Results}\label{sec:numerics}
\vspace{-0.1cm}

In this section, we illustrate the behavior of the Stackelberg game and validate the analysis. Unless otherwise stated, we use $q_{01}=0.5$, $q_{10}=1$, $q=0.55$, and $\eta=0.02$. The steady-state distribution of the source states is then $\pi_0=\tfrac{2}{3}$ and $\pi_1=\tfrac{1}{3}$, and the PC threshold is $q\pi_0+\eta=0.3867$. The default network size is $n=50$, the default sender budget is $R=20$, and the default gossip cap is $\Lambda=20$. In every figure, the continuous curves are evaluated from the backward recursion of Theorem~\ref{thm:backward_recursion} and the markers come from Monte Carlo simulations of the model of \cref{sec:model}.

We first fix the sender policy and vary the gossip rate. Fig.~\ref{fig:receiver-policy-regimes}(a) shows the receiver utility for $(s,c)=(17,3)$, which lies in the strategic regime $s>c$. The utility increases with $\lambda$ at every network size, as stated in Theorem~\ref{thm:UR-monotone}. The improvement grows with $n$, since a larger network receives fewer pushes per node and depends more on gossip to distribute the freshest packet. The figure also shows that PC feasibility is not a property of the policy alone. For $n=50$, the curve starts below the PC threshold and crosses it near $\lambda=10$, so the same policy is infeasible under weak gossip and feasible under strong gossip.

Fig.~\ref{fig:receiver-policy-regimes}(b) shows the mirror policy $(s,c)=(5,15)$, for which $c>s$ and Theorem~\ref{thm:UR-monotone} does not apply. For $n=2$ and $n=5$ the utility increases throughout, while for $n=50$ it dips before it recovers. In all three cases the derivative changes sign at most once and the maximum over $[0,\Lambda]$ is attained at $\lambda=\Lambda$, consistent with Conjecture~\ref{conj:one_dip}. At low gossip rates, nodes may replace accurate but older packets with fresher packets that have already become inaccurate following a source transition, which produces the initial dip. As $\lambda$ increases, corrective packets propagate more rapidly, and this gain eventually outweighs the initial loss.

\begin{figure}[t]
\centering
\subfloat[]{%
    \includegraphics[width=0.48\linewidth]{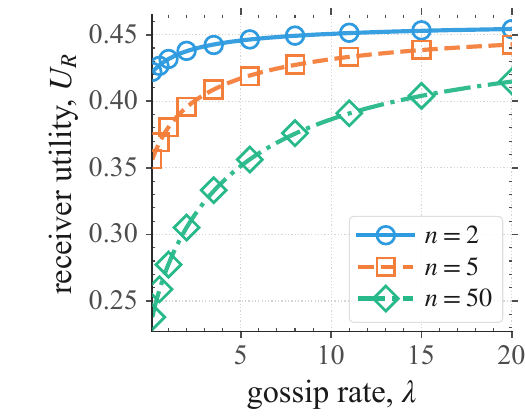}%
    \label{fig:receiver-strategic}%
}
\hfil
\subfloat[]{%
    \includegraphics[width=0.48\linewidth]{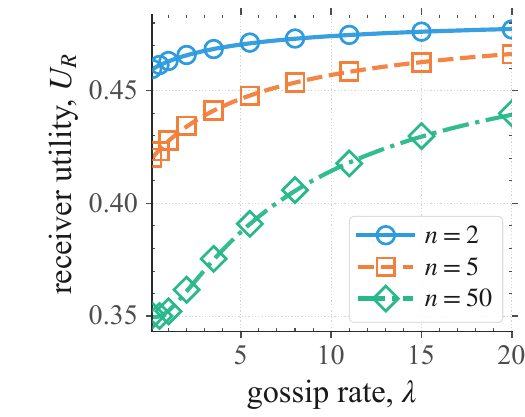}%
    \label{fig:receiver-distortive}%
}
\caption{Receiver utility $U_R$ against the gossip rate $\lambda$ at two fixed sender policies, (a) $(s,c)=(17,3)$ with $s>c$ and (b) $(s,c)=(5,15)$ with $c>s$.}
\label{fig:receiver-policy-regimes}
\end{figure}

Fig.~\ref{fig:fixed-policy-sender}(a) reports the sender utility for three policies that share the same budget. It decreases in $\lambda$ at $(17,3)$, stays at $\pi_1$ at $(10,10)$, and increases in $\lambda$ at $(5,15)$, which is the statement of Corollary~\ref{cor:direct-sign}. In both asymmetric cases gossip moves $\US$ toward the symmetric value $\pi_1$ without reaching it.

Fig.~\ref{fig:fixed-policy-sender}(b) shows which part of the network gossip repairs, at the strategic policy $(17,3)$ with $n=50$. The state-conditional accuracies plotted there are $\Pr(S_j=0\mid Q=0)=\tfrac{f_1^{(0)}}{\pi_0}$ and $\Pr(S_j=1\mid Q=1)=\tfrac{f_1^{(1)}}{\pi_1}$, together with the unconditional fraction $\Pr(S_j=1)$ of nodes declaring state~$1$, which equals $\US$ by~\eqref{eq:util-f1}. As $\lambda$ grows, the state-$0$ accuracy rises steeply, while the state-$1$ accuracy dips at small $\lambda$ and then rises slowly. The initial dip arises from the same freshness--accuracy mismatch as in Fig.~\ref{fig:receiver-policy-regimes}(b). Although both state-conditional accuracies eventually exceed their values at $\lambda=0$, the sender utility decreases because the improvement in the mode-$0$ accuracy, $f_1^{(0)}$, exceeds the corresponding improvement in the mode-$1$ accuracy, $f_1^{(1)}$. Gossip therefore primarily corrects the state that the sender updates less frequently, which is also the more probable state under $\pi_0=\tfrac{2}{3}$.

\begin{figure}[t]
\centering
\subfloat[]{%
    \includegraphics[width=0.48\linewidth]{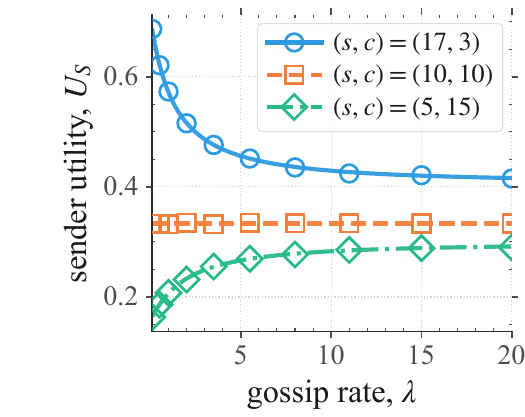}%
    \label{fig:sender-fixed-policy}%
}
\hfil
\subfloat[]{%
    \includegraphics[width=0.48\linewidth]{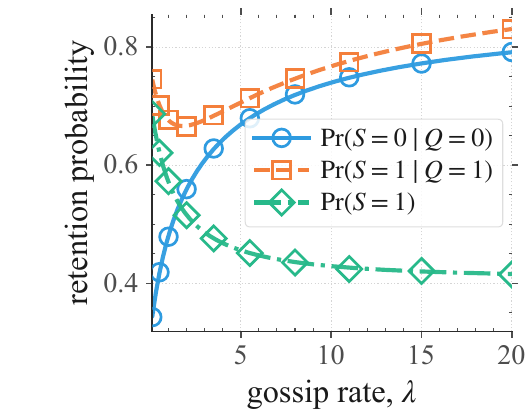}%
    \label{fig:retention}%
}
\caption{Fixed-policy effects of gossip with $n=50$. (a) Sender utility $U_S$ against $\lambda$ at three policies with $s+c=R$. (b) State-conditional accuracies and the fraction of nodes declaring state~$1$ against $\lambda$, at $(s,c)=(17,3)$.}
\label{fig:fixed-policy-sender}
\end{figure}

We now move to the budget line $s=R-c$ and to its strategic half $c<\tfrac{R}{2}$, where the sender operates at the smallest PC-feasible point $c_{\min}^{\mathrm{str}}(\lambda)$ of $\Slack(c,\lambda)$. Fig.~\ref{fig:equilibrium-budget} fixes $\Lambda=20$ and varies the sender budget, showing the equilibrium push rates in Fig.~\ref{fig:equilibrium-budget}(a) and the resulting utilities in Fig.~\ref{fig:equilibrium-budget}(b). Over the whole range $c^\star$ stays well below $\tfrac{R}{2}$, so $s^\star>c^\star$, the receiver's best response is $\lambda^\star=\Lambda$, and the equilibrium is the point $c^\star=c_{\min}^{\mathrm{str}}(\Lambda)$ of Theorem~\ref{thm:SE}. There $c^\star$ remains nearly constant while $s^\star$ increases almost linearly with $R$, and the receiver utility stays at the PC threshold, so the sender collects the entire gain, at a diminishing rate.

The last experiment isolates the effect of gossip on the equilibrium itself. In the strategic regime the receivers gossip at the cap, so the realized gossip rate is $\lambda^\star=\Lambda$, and we vary $\Lambda$ with $R=20$ and $n=50$ held fixed. Fig.~\ref{fig:equilibrium-lambda} puts the sender utility and the constraint that shapes it on one pair of axes. The right axis carries the boundary $c_{\min}^{\mathrm{str}}(\lambda^\star)$. It exists only above the critical gossip rate near $5.4$, where it approaches $\tfrac{R}{2}$ from below, and it decreases afterwards, as Corollary~\ref{cor:cmin-lambda-strategic} states. The equilibrium policy tracks the boundary, so $c^\star$ falls as the gossip rate grows and $s^\star=R-c^\star$ rises with it, and more gossip lets the sender operate at a more biased policy. The left axis carries the equilibrium sender utility, which increases in $\lambda^\star$ at a decreasing slope over the plotted range. This is the opposite of the fixed-policy behavior in Fig.~\ref{fig:fixed-policy-sender}(a), where more gossip lowers $\US$ at $(17,3)$.

The two frozen curves separate the effects that produce this reversal. Each is indexed by a reference gossip rate $\lambda_0$ and holds the policy at the equilibrium value $\theta^\star(\lambda_0)=(R-c_{\min}^{\mathrm{str}}(\lambda_0),\,c_{\min}^{\mathrm{str}}(\lambda_0))$ while the realized gossip rate is swept above $\lambda_0$. For each frozen strategic policy, the policy-adjustment term in~\eqref{eq:US-eq-decomp} vanishes, leaving only the negative direct effect established in Corollary~\ref{cor:direct-sign}. Consequently, both frozen-policy curves decrease with $\lambda$. A frozen curve meets the equilibrium curve at its own $\lambda_0$, where the frozen policy is the equilibrium policy, and falls away from it above.

At the equilibrium the PC binds, so $\UR(\theta^\star,\lambda)=q\pi_0+\eta$. Solving the receiver expression in~\eqref{eq:util-f1} for $f_1^{(0)}(\theta^\star,\lambda)$ and substituting it into the sender expression gives
\begin{equation}\label{eq:eq-US-identity}
    \US(\theta^\star,\lambda)=\frac{f_1^{(1)}(\theta^\star,\lambda)-\eta}{q}.
\end{equation}
The sender's equilibrium payoff therefore depends on the gossip rate only through the state-$1$ accuracy, and the rise in Fig.~\ref{fig:equilibrium-lambda} says that this accuracy grows with $\lambda^\star$. We evaluated the recursion at $2{,}000{,}000$ randomly drawn parameter sets, of which $923{,}550$ satisfy \cref{as:outside} and $362{,}191$ of those are PC-feasible at each of $128$ sampled gossip rates. In each of the latter, the equilibrium sender utility never fell between consecutive gossip rates on the strategic half of the budget line. The draws were independent, with $n$ uniform on $\{2,\ldots,300\}$, $q$ uniform on $[0.02,0.9]$, $q_{01}$ and $q_{10}$ logarithmically uniform on $[10^{-2},1]$, $\eta$ on $[10^{-4},3\!\times\!10^{-2}]$, and $R$ on $[10,10^{2}]$, which keeps the push rate per node comparable to the source rates.

\begin{figure}[t]
\centering
\subfloat[]{%
    \includegraphics[width=0.48\linewidth]{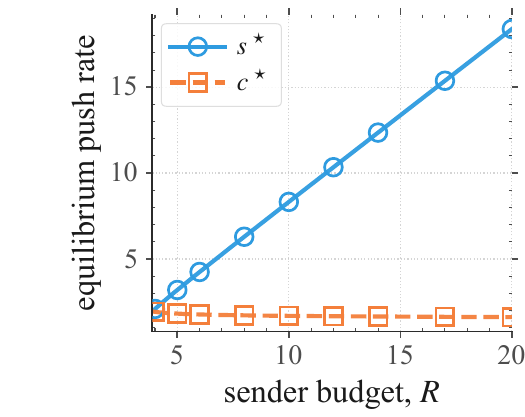}%
    \label{fig:eq-policies-budget}%
}
\hfil
\subfloat[]{%
    \includegraphics[width=0.48\linewidth]{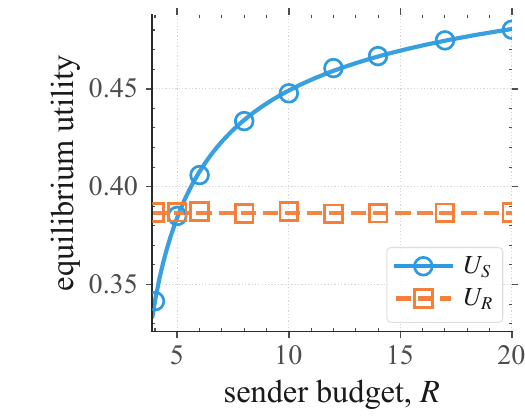}%
    \label{fig:eq-utilities-budget}%
}
\caption{Equilibrium comparative statics in the sender budget $R$, with $n=50$ and $\Lambda=20$. (a) Equilibrium push rates. (b) Equilibrium utilities, with the receiver utility at the PC threshold $q\pi_0+\eta$.}
\label{fig:equilibrium-budget}
\end{figure}

\begin{figure}[t]
\centering
\includegraphics[width=0.8\columnwidth]{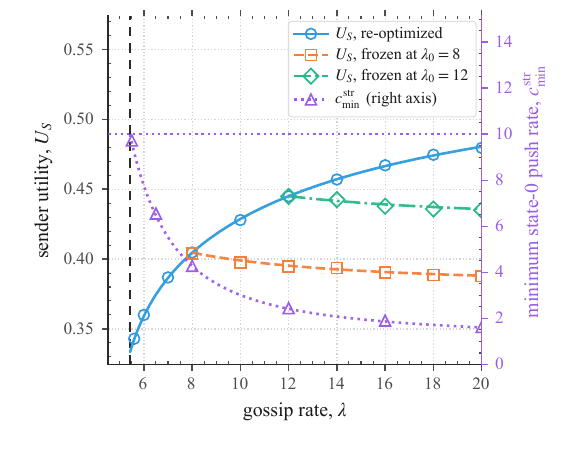} \vspace{-0.35cm}
\caption{Effect of the realized gossip rate $\lambda^\star\!=\!\Lambda$ on the equilibrium. The left axis carries the equilibrium sender utility and the two frozen-policy curves, each of which holds the policy at its equilibrium value $\theta^\star(\lambda_0)$ for a reference rate $\lambda_0$ while $\lambda^\star$ is swept above $\lambda_0$. The right axis carries $c_{\min}^{\mathrm{str}}$. The vertical dashed line marks the gossip rate below which $\mathcal F_{\lambda}^{\mathrm{str}}$ is empty, and the dotted horizontal line marks $\tfrac{R}{2}$ on the right axis.}
\label{fig:equilibrium-lambda}
\vspace{-0.2cm}
\end{figure}

In the regimes examined here, gossip helps the receivers and hurts the sender at a fixed strategic policy, while sender reoptimization returns the receiver utility to $q\pi_0\!+\!\eta$ and raises the sender payoff. Whether the second effect dominates the first at every admissible parameter remains open.

\section{Conclusion}\label{sec:conclusion}
\vspace{-0.1cm}

We studied strategic persuasion in timeliness-based gossip networks, where a sender observes a binary CTMC source and transmits state-dependent updates to a team of receivers who gossip among themselves. We modeled the interaction as a Stackelberg game and analyzed it through the SHS framework, which yields a backward recursion for the steady-state mode-tagged accuracies. We proved that, at a fixed gossip rate, the sender's budget constraint binds at interior optima. This reduces the sender's problem to a one-dimensional optimization on the budget line, where the sender utility is strictly decreasing in the state-$0$ push rate. We then proved that in the strategic regime $s > c$, the receiver utility is strictly increasing in the gossip rate and the sender utility is strictly decreasing, for all network sizes. Beyond this regime, we showed that $(c\!-\!s)\tfrac{\partial\US}{\partial\lambda}\! \ge\! 0$. Gossip therefore has no direction of its own, and it works against the asymmetry already present in the sender's policy rather than reinforcing it. We also proved that an optimistic Stackelberg equilibrium exists, and that it is unique and equal to the smallest PC-feasible policy on the budget line whenever that policy lies on the strategic half. Future work includes costly gossip, where receivers pay for information exchange, and network topologies beyond the fully connected case.

\section*{Acknowledgment}
\vspace{-0.1cm}
The authors used OpenAI’s GPT-5.6 Sol and Anthropic’s Claude Opus 5 for \LaTeX{} formatting, grammar correction, and generating the system model figure. Prompts included shortening technical explanations, improving lengthy derivations, and drawing the TikZ schematic. The authors reviewed and verified all generated content and take full responsibility for the publication.

\bibliographystyle{IEEEtran}
\bibliography{refs}

\end{document}